\documentclass[journal,draftcls,onecolumn,12pt,twoside]{IEEEtranTCOM} % IEEE transactions on communications

\IEEEoverridecommandlockouts
\usepackage{cite}
\usepackage{amsmath,amssymb,amsfonts,amsthm,mathrsfs}
\usepackage{arydshln}
\usepackage{algorithmic}
\usepackage{graphicx}
\usepackage{textcomp}
\usepackage{xcolor}
\usepackage{tikz}
\usetikzlibrary{shapes,arrows}
\usepackage{bm}

\def \Z {{\mathbb Z}}

\def \x {{\bf x}}
\def \y {{\bf y}}

\newtheorem{theorem}{Theorem}
\newtheorem{lemma}{Lemma}
\newtheorem{corollary}[theorem]{Corollary}

\newtheorem{remark}{Remark}
\newtheorem{note}[theorem]{Note}
\newtheorem{definition}{Definition}

\def\BibTeX{{\rm B\kern-.05em{\sc i\kern-.025em b}\kern-.08em
    T\kern-.1667em\lower.7ex\hbox{E}\kern-.125emX}}
\begin{document}

\title{DNA Codes over the Ring $\mathbb{Z}_4 + w\mathbb{Z}_4$}

\author{\IEEEauthorblockN{
Adel Alahmadi$^1$\footnote{$^1$ Department of Mathematics, King Abdulaziz University, Jeddah, Saudi Arabia}, 
Krishna Gopal Benerjee$^2$\footnote{$^2$ Department of Electrical Engineering, Indian Institute of Technology Kanpur, Uttar Pradesh - 208016, India}, 
Sourav Deb$^3$\footnote{$^3$ Dhirubhai Ambani Institute of Information and Communication Technology Gandhinagar, Gujarat - 382007, India}, and \\ 
Manish K Gupta$^3$,~\IEEEmembership{Senior Member,~IEEE}
\thanks{email: analahmadi@kau.edu.sa, kgopal@iitk.ac.in, sourav\_deb@daiict.ac.in, and mankg@computer.org }}
% \\ \IEEEauthorblockA{Authors affiliation \\
% Email: }
}

\maketitle

\begin{abstract}
In this present work, we generalize the study of construction of DNA codes over the rings $\mathcal{R}_\theta=\mathbb{Z}_4+w\mathbb{Z}_4$, $w^2 = \theta $ for $\theta \in \mathbb{Z}_4+w\mathbb{Z}_4$. 
Rigorous study along with characterization of the ring structures is presented. 
We extend the Gau map and Gau distance, defined in \cite{DKBG}, over all the $16$ rings $\mathcal{R}_\theta$. 
Furthermore, an isometry between the codes over the rings $\mathcal{R}_\theta$ and the analogous DNA codes is established in general.
Brief study of dual and self dual codes over the rings is given including the construction of special class of self dual codes that satisfy reverse and reverse-complement constraints. 
The technical contributions of this paper are twofold.
Considering the Generalized Gau distance, Sphere Packing-like bound, GV-like bound, Singleton like bound and Plotkin-like bound are established over the rings $\mathcal{R}_\theta$. 
In addition to this, optimal class of codes are provided with respect to Singleton-like bound and Plotkin-like bound. 
Moreover, the construction of family of DNA codes is proposed that satisfies reverse and reverse-complement constraints using the Reed-Muller type codes over the rings $\mathcal{R}_\theta$. 
\end{abstract}


\section{Introduction}
The topic of DNA codes has fascinated researchers for the last two decades. 
The intrinsic idea of computation using DNA has been posed in a seminal paper by Adleman in 1994 \cite{adleman1994molecular}. 
Driven by the ground-breaking idea, the branch of mathematical construction for artificial DNA codes and many more interdisciplinary topics, such as DNA-based data storage systems, computation using DNA tiles, etc. have emerged.
DNA (\textit{DeoxyriboNucleic Acid}) stores of all genetic as well as biological instructions of life. 
DNA strands can be recognised by four \textit{nucleotides (nt)} or bases, Adenine ($A$), Guanine ($G$), Cytosine ($C$) and Thymine ($T$), where the Watson-Crick complements for $A$ and $G$ are $T$ and $C$ respectively, $i.e.$, $A^c=T$ and $G^c=C$.
A DNA strand is a string of certain length containing the symbols $A$, $G$, $C$ and $T$, $e.g.$, $ACTTAGA$ is a DNA string of length $7$, as the string contains $7$ symbols, not necessarily distinct.
For a DNA string $\x = (x_1 \ x_2 \ \ldots \ x_n)$ of length $n$, the reverse and complement DNA strings are considered as $\x^r = (x_n \ x_{n-1} \ \ldots \ x_1)$ and $\x^c = (x^c_1 \ x^c_2 \ \ldots \ x^c_n)$, respectively.
Following these definitions, the reverse-complement DNA string of the DNA string $\x = (x_1 \ x_2 \ \ldots \ x_n)$ of length $n$ will be $\x^{rc} = (x^c_n \ x^c_{n-1} \ \ldots \ x^c_1)$.
Considering the previously stated DNA string $\x=ACTTAGA$, the corresponding reverse, complement and reverse-complement DNA strings will be $\x^r=AGATTCA$, $\x^c=TGAATCT$ and $\x^{rc}=TCTAAGT$, respectively.  
The key point of using DNA for computing is the DNA hybridization in between the string(s), which however is also the source of errors.
% One such important aspect in these context is to ensure that secondary like structure must be avoided during the computation \cite{milenkovic2005dna}. 
The purpose of constructing DNA codes using different approaches is to generate sufficiently dissimilar DNA strands to avoid error(s) in any aspect \cite{gaborit2005linear,king2003bounds}. 
Meanwhile, various celebrated constructions have been proposed \cite{limbachiya2016art}, which produce optimal DNA codes satisfying different constraints, where DNA code construction using rings and fields is treated as a prominent one.
Recently, DNA codes satisfying multiple constraints are constructed in \cite{Benerjee2020,9214530}.

% In this work, we expand the construction of DNA codes using the ring $\mathbb{Z}_4+w\mathbb{Z}_4$ proposed by Limbachiya \textit{et. al} \cite{DKBG}.
We denote the DNA symbol set $\Sigma_{DNA}= \{A,G,C,T\}$. 
Also Hamming distance $d_H$ between two DNA strings $\x=(x_1 \ x_2 \ \ldots \ x_n)$ and $\y=(y_1 \ y_2 \ \ldots \ y_n)$ of length $n$ is denoted as $d_H(\x,\y)$ to be the number of positions where the strings $\x$ and $\y$ differ, $i.e.$, the number of elements in the set $\{i:x_i \neq y_i\}$, where $i \in \{1,2, \ldots n\}$. 
A DNA code over the symbol set $\Sigma_{DNA}$ is defined by $\mathscr{C}_{DNA} (n,M,d_H) \subseteq \Sigma^n_{DNA}$, containing $M$ distinct DNA strings or codewords, each of length $n$ and having minimum Hamming distance $d_H$ between any two codewords.     
According to the definition $d_H=\min\{d_H (\x,\y):\mbox{for all } \x,\y \in \mathscr{C}_{DNA}, s.t., \x \neq \y \}$.
In this paper, the DNA code $\mathscr{C}_{DNA}$ is considered as reversible DNA code or closed under reverse constraint, if for each $\x\in\mathscr{C}_{DNA}$, $\x^r\in\mathscr{C}_{DNA}$.
Similarly, a DNA code $\mathscr{C}_{DNA}$ is a complement DNA code if, for each $\x\in\mathscr{C}_{DNA}$, $\x^c\in\mathscr{C}_{DNA}$.
We also consider that if, for each $\x\in\mathscr{C}_{DNA}$, $\x^{rc}\in\mathscr{C}_{DNA}$ then the DNA code $\mathscr{C}_{DNA}$ satisfies the reverse-complement constraint, and the code is called reversible-complement DNA code.
To maximize the output reliability of DNA codes, a set of combinatorial constraints, $i.e.$, Hamming constraint, reverse constraint, reverse-complement constraint and $GC$-content constraint, are taken into consideration in efficient constructions \cite{marathe2001combinatorial}.
Thus motivated, researchers are attracted towards the algebraic construction of DNA codes where each codeword satisfies the given constraints to ensure the appropriate match during the process of DNA hybridization and avoid error(s). 

Rykov \textit{et. al} \cite{rykov2001dna} pioneered the construction of DNA codes using sets of reverse-complement cyclic codes over quaternary alphabets designating $0,1,2$ and $3$ for $A,C,G$ and $T$, respectively.
Theoretical bounds are also established for parameters of the constructed codes in this paper. 
From the theoretical point of view, it is quite evident to construct DNA codes over the ring $\mathbb{Z}_4$ and the field $GF(4)$ or its variants.
The construction of linear DNA codes using the ring $\mathbb{Z}_4$ was proposed by King \textit{et. al} in \cite{gaborit2005linear}, where the constructed codes follow the reverse and reverse-complement constraints.
A parallel approach can be found in \cite{abualrub2006construction} to construct linear and additive cyclic codes of odd length that satisfy reverse-complement constraints over the Galois filed $GF(4)$ of the purpose of DNA computing. 
In \cite{faria2010dna,rocha2010dna}, two different approaches can be observed to construct DNA sequences using BCH codes over $GF(4)$ and $\mathbb{Z}_4$-linear codes, respectively. 

In this paper, we restrict our focuses to the ring $\mathbb{Z}_4+w\mathbb{Z}_4$ and its variants, containing $16$ elements.
The class of finite rings $S_{2m}=\mathbb{Z}_{2m}+i\mathbb{Z}_{2m}$, $i^2=-1$ was introduced in 2005 by Choie and Dougherty in \cite{choie2005codes} and self dual codes are studies.
Motivated by the approach, in 2014 linear codes as well as self dual codes which are analogous allowing some projections from the ring $\mathbb{Z}_4+u\mathbb{Z}_4$ to $\mathbb{Z}_4$ and $\mathbb{F}_2+u\mathbb{F}_2$,  have been constructed by Yildiz \textit{et. al} \cite{yildiz2014linear}. 
Furthermore, MacWilliams identities were also established in this paper. 
% Construction of cyclic DNA codes over the ring $\mathbb{Z}_4+w\mathbb{Z}_4$ or its variants \textbf{has enriched the literature over the times.}
Cyclic DNA codes of odd length over the rings $\mathbb{Z}_4+w\mathbb{Z}_4$, $w^2=2$ and $\mathbb{Z}_4+u\mathbb{Z}_4$, $u^2=0$ have been constructed in \cite{dertli2016cyclic}, \cite{pattanayak2015cyclic}. 
Usage of the ring $\mathbb{F}_4+u\mathbb{F}_4$, $u^2=0$ in constructing reversible cyclic DNA codes can be found in \cite{srinivasulu2015reversible}. 
In \cite{bayram2015codes}, Siap \textit{et. al} presented an extensive study on characterizing the structure of linear, constacyclic and cyclic codes over the ring $\mathbb{F}_4[v]/(v^2-v)$ and construction of optimal DNA codes satisfying reverse-complement constraint via reversible codes over the considered ring was proposed.
Ding \textit{et. al} in \cite{dinh2019construction}, have constructed cyclic DNA codes over the ring $\mathbb{Z}_4[u]/(u^2-1)$ of odd length based on deletion distance.
In \cite{selfdualbandi2015}, self dual codes are constructed over the rings $\mathbb{Z}_4+w\mathbb{Z}_4$, $w^2=i$ where $i \in \{1, 2w\}$. 
Circulant self-dual codes and Type II codes over these rings were also studied in this paper. 
Recently in 2018, Limbachiya $et.\ al$ \cite{DKBG} defined the Gau distance on the elements of the ring $\mathbb{Z}_4+w\mathbb{Z}_4$, $w^2=2+2w$ and constructed DNA codes of even length, which satisfy the reverse and reverse-complement constraints. 
The authors also proposed the Gau map, which is an isometry between the ring elements and their corresponding DNA string of length two. 
That is, the map preserves the distance from the codes over the ring based on Gau distance and the corresponding DNA codes over $\Sigma^{2n}_{DNA}$ based on Hamming distance.

\subsection{Contributions:}
Driven by the construction  proposed in \cite{DKBG}, we extend the study of codes over the ring $\mathcal{R}_{\theta}=\mathbb{Z}_4+w\mathbb{Z}_4$, $w^2=\theta$ such that $\theta \in \mathbb{Z}_4+w\mathbb{Z}_4$ with respect to the Generalized Gau distance and establish a complete characterization of DNA codes using Generalized Gau map, where the Generalized Gau distance and Generalized Gau map are Gau distance and Gau map defined over the ring $\mathbb{Z}_4+w\mathbb{Z}_4$ for any $\theta\in\mathbb{Z}_4+w\mathbb{Z}_4$.
The ring structure classification provides a symmetric distribution of $8$ chain and $8$ non-chain rings.
The generalized construction proposes the necessary and sufficient condition so that the DNA codes satisfy the reverse and reverse-complement constraints.
Using the construction, sufficient conditions on different categories of self dual codes that satisfy the reverse and reverse-complement constraints are proposed. 
The present work proposes theoretical bounds on parameters for the construction.
Furthermore, we obtain the parameters of different classes of optimal codes that satisfy Singleton-like bound and Plotkin-like bound.
A family of DNA codes satisfying reverse and reverse-complement constraints has been proposed using Reed-Muller type of codes over the ring.
  
\subsection{Organization:}
This paper is organized in the following manner Section \ref{sec:Preliminaries} provides a complete overview of the ring structures $\mathcal{R}_\theta$, $\theta \in \mathbb{Z}_4+w\mathbb{Z}_4$. 
The generalization of Gau map and Gau distance over the ring $\mathcal{R}_\theta$ is proposed in Section \ref{sec:Gau}.
The necessary and sufficient conditions for reverse and complement codes are derived and exploited to construct DNA codes using Reed-Muller type of codes. 
In Section \ref{sec:dual and selfdual}, we discuss the dual and self dual codes over the ring $\mathcal{R}_\theta$ along with the necessary conditions such that a special class of constructed codes satisfies reverse and complement constraints, which can be considered as a rich source of efficient DNA codes.
Section \ref{sec:Bounds} includes theoretical bounds that are established considering the generalized Gau distance over the rings $\mathcal{R}_\theta$ along with the parameters of families of optimal codes corresponding to certain bounds.
Section \ref{sec:Conclusion} concludes the paper.

\section{Preliminaries on the rings $\mathcal{R}_{\theta}, \ \text{where} \ \theta \in \mathbb{Z}_4+ w \mathbb{Z}_4$}
\label{sec:Preliminaries}
% \textbf{Need to arrange in flow and need to cite \cite{DKBG}} 

Throughout the paper, we denote $\mathcal{R}_{\theta}= \mathbb{Z}_4+ w \mathbb{Z}_4$ where $w^2 = \theta\in\mathbb{Z}_4+ w \mathbb{Z}_4$.
There are $16$ such rings. 
% , since the number of elements of the ring $\mathcal{R}$ is $16$. 
These $16$ rings can be classified into two types, $i.e.$, chain rings and non-chain rings due to their different ring structures. 
One can observe that, the division into two classes is symmetric, $i.e.$, $8$ of these rings are chain rings and the rest are non-chain rings.
An explicit classification is given in Table \ref{class} for the rings $\mathcal{R}_{\theta}, \theta \in \mathbb{Z}_4 + w\mathbb{Z}_4$.

\begin{table}[ht]
\centering
\caption{The classification for the rings $\mathcal{R}_{\theta}, w^2= \theta$} 
\begin{tabular}{|l|l|}
\hline
Chain Rings & $\theta \in \{2, 3, 1+w, 3+w,1+2w, 2+2w, 1+3w, 3+3w\}$ \\ \hline 
Non-Chain Rings    &  $\theta \in \{0, 1, w, 2w, 3w, 2+w, 3+2w, 2+3w\}$  \\ \hline 
\end{tabular}
\label{class}
\end{table}

% We have considered only the class of chain rings in this paper.
% For the rest of the paper, $\mathcal{R}_\theta$ denotes a chain ring for a fixed $\theta$. 
We discuss the structure of the $8$ finite commutative principal local chain rings $\mathcal{R}_{\theta}, \theta \in \mathbb{Z}_4+ w \mathbb{Z}_4$ in this section.
For given $\theta$, the set of zero divisors and the set of units of the ring $\mathcal{R}_\theta$ are denoted as $\mathcal{Z}_{\theta}$ and $\mathcal{U}_{\theta}$, respectively.
The set $\mathcal{R}^n_{\theta}$, along with the operations additions ``$+$" and multiplication ``$.$" can be viewed as a module over the ring $\mathcal{R}_{\theta}$. 
Some further classifications can be done in the class of chain rings in terms of basic ring structures.

\subsection{Ring $\mathcal{R}_{\theta}, \theta \in \{1+w, 3+w, 1+3w, 3+3w \}$.}

For $\mathcal{R}_{\theta}, \theta \in \{a+wb : a, b \in 2 \mathbb{Z}_4 +1 \} = \{1+w, 3+w, 1+3w, 3+3w \}$, we have,
$\mathcal{Z}_{1+w} = \mathcal{Z}_{3+w} = \mathcal{Z}_{1+3w} = \mathcal{Z}_{3+3w} = \{a+wb : a, b \in 2 \mathbb{Z}_4 \} = \{0, 2, 2w, 2+2w \}$; and
$\mathcal{U}_{1+w} = \mathcal{U}_{3+w} = \mathcal{U}_{1+3w} = \mathcal{U}_{3+3w} = \{a+wb : a, b \in \mathbb{Z}_4, a+b \in 2\mathbb{Z}_4+1 \} \cup \{a+wb : a, b \in 2\mathbb{Z}_4+1, a+b \in 2\mathbb{Z}_4 \} = \{1, 3, w, 3w, 1+w, 2+w, 3+w, 1+2w, 3+2w, 1+3w, 2+3w, 3+3w \}$.

The rings satisfy the chain condition with $3$ distinct ideals as, $\langle 0 \rangle \subset \langle 2 \rangle = \langle 2w \rangle = \langle 2+2w \rangle \subset \mathcal{R}_{\theta}$.

\subsection{Ring $\mathcal{R}_{\theta}, \theta \in \{2, 2+2w \}$.}

For $\mathcal{R}_{\theta}, \theta \in \{a+wb : a=2, b \in 2 \mathbb{Z}_4 \} = \{2, 2+2w \}$, we have,
$\mathcal{Z}_{2} = \mathcal{Z}_{2+2w} = \{a+wb : a \in 2 \mathbb{Z}_4, b \in \mathbb{Z}_4 \} = \{0, 2,w, 2w, 3w, 2+w, 2+2w, 2+3w \}$; and
$\mathcal{U}_{2} = \mathcal{U}_{2+2w} = \{a+wb : a \in 2\mathbb{Z}_4+1, b \in \mathbb{Z}_4 \} = \{1, 3, 1+w, 3+w, 1+2w, 3+2w, 1+3w, 3+3w \}$. 

These rings have $5$ distinct ideals as, $\langle 0 \rangle \subset \langle 2w \rangle \subset \langle 2 \rangle = \langle 2+2w \rangle \subset \langle w \rangle = \langle 3w \rangle = \langle 2+w \rangle = \langle 2+3w \rangle \subset \mathcal{R}_{\theta}$.

\subsection{Ring $\mathcal{R}_{\theta}, \theta \in \{3, 1+2w \}$.}

For $\mathcal{R}_{\theta}, \theta \in \{a+wb : a \in 2\mathbb{Z}_4+1, b \in 2 \mathbb{Z}_4, a \neq b+1 \} = \{3, 1+2w \}$, We have, 
$\mathcal{Z}_{3} = \mathcal{Z}_{1+2w} = \{a+wb : a,b \in \mathbb{Z}_4, a-b \in 2\mathbb{Z}_4 \} = \{0, 2, 2w, 1+w, 3+w, 2+2w, 1+3w, 3+3w \}$; and
$\mathcal{U}_{3} = \mathcal{U}_{1+2w} = \{a+wb : a,b \in \mathbb{Z}_4, a-b \in 2\mathbb{Z}_4+1 \} = \{1, 3, w, 3w, 2+w, 1+2w, 3+2w, 2+3w \}$.

The rings have $5$ distinct ideals as, $\langle 0 \rangle \subset \langle 2+2w \rangle \subset \langle 2 \rangle = \langle 2w \rangle \subset \langle 1+w \rangle = \langle 3+w \rangle = \langle 1+3w \rangle = \langle 3+3w \rangle \subset \mathcal{R}_{\theta}$.\\

Any subset of $\mathcal{R}^n_{\theta}$ is defined as a code $\mathscr{C}_{\theta} (n,M,d)$ over the ring $\mathcal{R}_{\theta}$ (for given $\theta$), which consists of $M$ distinct codewords, each of length $n$ and any two distinct codewords differ by distance at least $d$. 
Any submodule of the module $\mathcal{R}^n_{\theta}$ is considered as a linear code $\mathscr{C}_{\theta}$ over the ring $\mathcal{R}_{\theta}$. 
We denote the generating matrix of the linear code $\mathscr{C}_{\theta}$ over the ring $\mathcal{R}_{\theta}$ by $G_{\theta}$ and the row span of $G_{\theta}$ is denoted as $\langle G_{\theta} \rangle$.
The generator matrix of the linear code $\mathscr{C}_{\theta}$ over the ring $\mathcal{R}_{\theta}$ in standard form, is given by:

\begin{itemize}
    \item[a)] For $\theta \in \{1+w, 3+w, 1+3w, 3+3w \}$, the generating matrix $G_{\theta}$ in standard form, of the linear code $\mathscr{C}_{\theta}$ over the ring $\mathcal{R}_{\theta}$ will be,  
\begin{equation}
G_{\theta}=\left(
\begin{array}{ccc}
I_{k_0} & A_{0,1}   & A_{0,2}    \\
0       & 2I_{k_1} & 2A_{1,2}    \\
\end{array} \right),
\label{G_1}
\end{equation} 
where the matrices $A_{i,j}$ are defined over the ring $\mathcal{R}_{\theta}$ for $0 \leq i < j \leq 2$.
The matrix $G_\theta$ is called type $\{k_0,k_1\}$, and the linear code with the generator matrix $G_\theta$ is called type $\{k_0,k_1\}$ linear code.
The total number of codewords of the type $\{k_0,k_1\}$ linear code $\mathscr{C}_{\theta}$ will be $16^{k_0}4^{k_1}$, where $\theta \in \{1+w, 3+w, 1+3w, 3+3w \}$.
    \item[b)] For $\theta \in \{2, 2+2w \}$, the generating matrix $G_{\theta}$ in standard form of the linear code $\mathscr{C}_{\theta}$ over the ring $\mathcal{R}_{\theta}$ will be, 
\begin{equation}
G_{\theta}=\left(
\begin{array}{ccccc}
I_{k_0} & A_{0,1}   & A_{0,2}   & A_{0,3}   & A_{0,4}  \\
0       & wI_{k_1} & wA_{1,2} & wA_{1,3} & wA_{1,4} \\
0       & 0         & 2I_{k_2}  & 2A_{2,3}  & 2A_{2,4} \\
0       & 0         & 0         & 2wI_{k_3}  & 2wA_{3,4} \\
\end{array} \right),
\label{G_2}
\end{equation} 
where the matrices $A_{i,j}$'s are defined over the ring $\mathcal{R}_{\theta}$ as, $A_{i,j} = B^1_{i,j} + wB^2_{i,j} + 2B^3_{i,j} + 2wB^4_{i,j}$, where $\ B^k_{i,j}$ are binary matrices for $0 \leq i < j \leq 4$ and $k= 1,2,3,4$.
The linear code $\mathscr{C}_{\theta}$ with the generating matrix in the given form, is of type $\{k_0, k_1, k_2, k_3\}$ and has   $16^{k_0}8^{k_1}4^{k_2}2^{k_3}$ codewords in total, where $\theta \in \{2, 2+2w \}$.
    \item[c)] For $\theta \in \{3, 1+2w \}$, the generating matrix $G_{\theta}$ of the linear code $\mathscr{C}_{\theta}$ over the ring $\mathcal{R}_{\theta}$, in standard form will be the matrix given in (\ref{G_3}), where the matrices $A_{i,j}$'s are defined over the ring $\mathcal{R}_{\theta}$ as, $A_{i,j} = B^1_{i,j} + (1+w)B^2_{i,j} + 2B^3_{i,j} + (2+2w)B^4_{i,j}$, where $\ B^k_{i,j}$ are binary matrices for $0 \leq i < j \leq 4$ and $k= 1,2,3,4$.
From the generating matrix $G_{\theta}$ in the standard form, it can deduced that the linear code $\mathscr{C}_{\theta}$ is of type $\{k_0, k_1, k_2, k_3\}$ and has $16^{k_0}8^{k_1}4^{k_2}2^{k_3}$ codewords, where $\theta \in \{3, 1+2w \}$. 
\begin{equation}
G_{\theta}=\left(
\begin{array}{ccccc}
I_{k_0} & A_{0,1}   & A_{0,2}   & A_{0,3}   & A_{0,4}  \\
0  & (1+w)I_{k_1} & (1+w)A_{1,2} & (1+w)A_{1,3} & (1+w)A_{1,4} \\
0       & 0         & 2I_{k_2}  & 2A_{2,3}  & 2A_{2,4} \\
0       & 0         & 0         & (2+2w)I_{k_3}  & (2+2w)A_{3,4} \\
\end{array} \right)
\label{G_3}
\end{equation} 
\end{itemize}

\section{Generalized Gau Map and Generalized Gau Distance}
\label{sec:Gau}

We define a correspondence between the elements of the ring $\mathcal{R}_\theta$ and the DNA alphabets for the purpose of constructing DNA codes satisfying different constraints.
The correspondence leads to an isometry (Distance preserving map) between the codes over the ring $\mathcal{R}_\theta$ and the DNA codes in parallel.
In the case of DNA codes, the Hamming Distance has been defined, while for the codes over the ring $\mathcal{R}_\theta$, Generalized Gau Distance $d_{G(\theta)}$ over $\mathcal{R}_\theta$ (eventually over $\mathcal{R}^n _\theta$) has been introduced.
The isometry has been established simultaneously by defining Generalized Gau Map $\phi$ over the ring $\mathcal{R}_\theta$.
One can verify that $\phi$ is a bijective mapping between the elements of the ring $\mathcal{R}_\theta$ and the pair DNA nucleotides and it will be independent of the choice of $\theta$.
The arrangement of matrix given in (\ref{Gau Distance Matrix}) is managed in such a way that the Hamming distance between any two distinct pair of DNA neucleotides in different rows or different columns is $2$, otherwise it is $1$ and we fill the elements of the matrix using the elements of the ring $\mathcal{R}_\theta$ according to the conditions $\phi^{-1}(\phi(x)^r)=3x$ and $\phi^{-1}(\phi(x)^c)=x+\lambda$, where $\lambda \in \{2, 2w, 2+2w\}$.  
Using the conditions we obtain the following restrictions over the elements of the matrix given in (\ref{Gau Distance Matrix}).

\begin{enumerate}
    \item $a_{1,1},a_{2,2}\in\{x \in \mathcal{R}_\theta: 2x = 0\}$ such that $a_{1,1}\neq a_{2,2}$ and $a_{1,1}\neq a_{2,2}+\lambda$, 
    \item $a_{1,4},a_{2,3}\in\{x \in \mathcal{R}_\theta: 3x = x+ \lambda \}$ such that $a_{1,4}\neq a_{2,3}$ and $a_{1,4}\neq a_{2,3}+ \lambda$, and
    \item $a_{1,2},a_{1,3}\in\{x \in \mathcal{R}_\theta: x \neq 3x, 3x = x+ \lambda\}$ such that $a_{1,2}\neq a_{1,3}$, $a_{1,2}\neq a_{1,3}+\lambda$, $a_{1,2}\neq3a_{1,3}$, $a_{1,2}\neq 3a_{1,3}+\lambda$.
\end{enumerate}

We construct the following example of a specific Generalized Gau Map $\phi$, using the abovesaid restrictions, where $\lambda = 2+2w$ is considered.

\begin{table}[ht]
\centering
\caption{An example of $\phi$: $\mathcal{R}_\theta\rightarrow \Sigma_{DNA}^2$ is illustrated.}
\begin{tabular}{|c|c||c|c|}
\hline
Ring element & DNA image &   Ring element & DNA image \\ 
$x$    &  $\phi(x)$  & $x$     & $\phi(x)$   \\
\hline \hline 
$2+2w$   & $AA$        & $w$  & $AC$  \\ 
\hline
$0$  & $TT$ &  $3w$ & $CA$ \\
\hline
$2w$     &  $GG$       & $2+3w$  & $GT$       \\
\hline
$2$   & $CC$ & $2+w$ & $TG$ \\ 
\hline
$1$     &  $AG$       & $1+w$     & $AT$  \\ 
\hline
$3$  & $GA$ & $3+3w$ & $TA$ \\
\hline
 $3+2w$    &  $CT$       & $3+w$    & $GC$  \\ 
\hline
$1+2w$    & $TC$ & $1+3w$ & $CG$ \\
\hline
\end{tabular}
\label{Gau_Map}
\end{table}

For a similar map, readers are referred to Table I, \cite{DKBG}.
One can observe that, different Generalized Gau Maps can be constructed by swapping the elements in the matrix (\ref{Gau Distance Matrix}), according to the given conditions.       
This leads us to enumerate the total number of Generalized Gau map, which can be constructed in this approach and we have following result.

\begin{remark}
For a given $\theta\in\mathbb{Z}_4+w \mathbb{Z}_4$, $2^{11}$ distinct Generalized Gau mappings $\phi:\mathcal{R}_\theta \rightarrow \Sigma^2_{DNA}$ exist such that $\phi^{-1}(\phi(x)^r)=3x$ and $\phi^{-1}(\phi(x)^c)=x+\lambda$ where $\lambda \in \{2, 2w, 2+2w\}$.
\end{remark}

\begin{note}
In this work, we have considered $\lambda = 2+2w$.
\end{note}
% \begin{equation}
% G=\left(
% \begin{array}{ccccc}
% I_{k_0} & A_{0,1}   & A_{0,2}   & A_{0,3}   & A_{0,4}  \\
% 0       & wI_{k_1} & wA_{1,2} & wA_{1,3} & wA_{1,4} \\
% 0       & 0         & 2I_{k_2}  & 2A_{2,3}  & 2A_{2,4} \\
% 0       & 0         & 0         & 2wI_{k_3}  & 2wA_{3,4} \\
% \end{array} \right),
% \end{equation} 

Thus motivated, we define the Generalized Gau Distance $d_{G(\theta)}$ over $\mathcal{R}_\theta$ as,
\begin{equation}
\begin{array}{cc}
   \mathscr{M} = &  \begin{array}{cc}
         & 
\begin{array}{lcccccccccccccccccc}
            A & & & & G & & & & & & C & & & & & & T & &
        \end{array}  
        \\
        \begin{array}{c}
           A \\
            G \\
            C \\
            T 
        \end{array}  
        &
        \left(\begin{array}{cccc}
         a_{1,1} & a_{1,2} & a_{1,3} & a_{1,4}  \\
        3a_{1,2} & a_{2,2} & a_{2,3} & 3a_{1,3}+2+2w     \\
        3a_{1,3} & 3a_{2,3} & a_{2,2}+2+2w & 3a_{1,2}+2+2w \\
        3a_{1,4} & a_{1,3}+2+2w & a_{1,2}+2+2w & a_{1,1}+2+2w
\end{array}\right)
    \end{array}
\end{array}
\label{Gau Distance Matrix}
\end{equation}

For $x, y \in \mathcal{R}_\theta$, from Matrix (\ref{Gau Distance Matrix}), let $x = a_{i,j}$ and $y = a_{i^{'},j^{'}}$ are the elements of the matrix $\mathscr{M}$ for some $0 \leq i,j \leq 3$ and $0 \leq i^{'},j^{'} \leq 3 $ then, the Generalized Gau distance $d_{G(\theta)}$ can be defined as 
\begin{equation} 
d_{G(\theta)}(x, y) = \min\{1, i + 3i^{'}\} + \min\{1, j + 3j^{'}\},
\end{equation}
where sum of the indices is calculated over $\mathbb{Z}_4$.
For any $\textbf{x} =(x_1 \  x_2 \ldots x_n)$ and $\textbf{y} =(y_1\ y_2\ldots y_n)$ in $\mathcal{R}_\theta^n$, the Generalized Gau distance $d_{G(\theta)}(\textbf{x}, \textbf{y}) = \sum_{i=1}^n d_{G(\theta)}(x_i, y_i)$ is a metric on $\mathcal{R}_\theta^n$ induced by the metric on the elements of the ring $\mathcal{R}_\theta$. 
The same notation $d_{G(\theta)}$ is used for both the metrics on $\mathcal{R}_\theta$ and $\mathcal{R}_\theta^n$. 
For a linear code $\mathscr{C}_\theta$ on $\mathcal{R}_\theta$, one can define a minimum Generalized Gau distance  $d_{G(\theta)}$ = $\min\{d_{G(\theta)}(\textbf{x}, \textbf{y}): \textbf{x}, \textbf{y} \in \mathscr{C}_\theta$ and $\textbf{x} \neq \textbf{y} \}$. 

Using computation, one can verify that, $d_{G(\theta)}$ is a metric on the ring $\mathcal{R}_\theta$ for any $\theta\in\mathcal{R}_\theta$.
Combining the Generalized Gau Map $\phi$ and the Generalized Gau Distance $d_{G(\theta)}$, we have the following theorem.

\begin{theorem}
The Generalized Gau Map $\phi:(\mathcal{R}^n_\theta,d_{G(\theta)}) \rightarrow (\Sigma^{2n}_{DNA}, d_H)$ is an isometry, $i.e.$, a distance preserving map. 
\end{theorem}

\begin{proof}
We will proof the result using Mathematical Induction over the parameter $n$. 

Base case: For any given $\theta\in\Z_4+w\Z_4$, the Generalized Gau Map $\phi:(\mathcal{R}_\theta,d_{G(\theta)}) \rightarrow (\Sigma^2_{DNA}, d_H)$ is a distance preserving map, $i.e.$, 
\[
d_{Gau}(x,y) = d_H(\phi(x),\phi(y))\mbox{ for any }x,y\in\mathcal{R}_\theta.
\]
Hypothesis: For given positive integer $m$, assume that the Generalized Gau Map $\phi:(\mathcal{R}^m_\theta,d_{G(\theta)}) \rightarrow (\Sigma^{2m}_{DNA}, d_H)$ is a distance preserving map, $i.e.$,
\[
d_{Gau}(\x,\y) = d_H(\phi(\x),\phi(\y))\mbox{ for any }\x,\y\in\mathcal{R}^m_\theta.
\]
Inductive step: For any $\x',\y'\in\mathcal{R}^{m+1}_\theta$, consider $\x,\y\in\mathcal{R}^m_\theta$ and $x_{m+1},y_{m+1}\in\mathcal{R}_\theta$ such that $\x'$ = $(\x\ x_{m+1})$ and $\y'$ = $(\y\ y_{m+1})$, where $\x$ = $(x_1\ x_2\ \ldots\ x_m)$ and $\y$ = $(y_1\ y_2\ \ldots\ y_m)$. 
Then, $\x'$ = $(x_1\ x_2\ \ldots\ x_m\ x_{m+1})$ and $\y'$ = $(y_1\ y_2\ \ldots\ y_m\ y_{m+1})$.
Now, 
\begin{equation*}
    \begin{split}
        d_{Gau}(\x',\y')  = & \sum_{i=1}^{m+1}d_{Gau}(x_i,y_i) \\
                        = & \left(\sum_{i=1}^md_{Gau}(x_i,y_i)\right)+d_{Gau}(x_{m+1},y_{m+1}) \\
                        = & d_{Gau}(\x,\y)+d_{Gau}(x_{m+1},y_{m+1}) \\
                        = & d_H(\phi(\x),\phi(\y))+d_H(\phi(x_{m+1}),\phi(y_{m+1})) \\
                        = & d_H(\x',\y').
    \end{split}
\end{equation*}
It follows the result.
\end{proof}

\begin{remark}
For any $\textbf{x} \in \mathscr{C}_\theta$ over the ring $\mathcal{R}_\theta$, $\phi^{-1}(\phi(\textbf{x})^r) \in  \mathscr{C}_\theta$ if and only if the DNA code $\phi(\mathscr{C}_\theta)$ is reversible. 
\label{reverse_c}
\end{remark}
\begin{remark}
For any $\textbf{x} \in \mathscr{C}_\theta$, $\phi^{-1}(\phi(\textbf{x})^c) \in  \mathscr{C}_\theta$ if and only if the DNA code $\phi(\mathscr{C}_\theta)$ is complement.
\label{complement_c}
\end{remark}
\begin{remark}
For any $\textbf{x} \in \phi(\mathscr{C}_\theta)$, if $\textbf{x}^r \in \phi(\mathscr{C}_\theta)$ and $\textbf{x}^c \in \phi(\mathscr{C}_\theta)$ then $\textbf{x}^{rc} \in \phi(\mathscr{C}_\theta) $. 
\label{r_rc_exist}
\end{remark}
\begin{remark}
for any $a$ and $b$ in $\mathcal{R}_\theta$, and any $\x$ = $(x_1\ x_2\ \ldots\ x_n)$ and $\y$ = $(y_1\ y_2\ \ldots\ y_n)$ in $\mathcal{R}^n_\theta$, consider
\begin{equation*}
    \begin{split}
        (a\x+b\y)^r = & \left(a(x_1\ x_2\ \ldots\ x_n)+b(y_1\ y_2\ \ldots\ y_n)\right)^r \\
                    = & \left((ax_1\ ax_2\ \ldots\ ax_n)+(by_1\ by_2\ \ldots\ by_n)\right)^r \\
                    = & (ax_1+by_1\ ax_2+by_2\ \ldots\ ax_n+by_n)^r \\
                    = & (ax_n+by_n\ ax_{n-1}+by_{n-1}\ \ldots\ ax_1+by_1) \\
                    = & (ax_n\ ax_{n-1}\ \ldots\ ax_1)+(by_n\ by_{n-1}\ \ldots\ by_1) \\
                    = & (a(x_n\ x_{n-1}\ \ldots\ x_1)+b(y_n\ y_{n-1}\ \ldots\ y_1) \\
                    = & a\x^r+b\y^r.
    \end{split}
\end{equation*}
Using similar approach, one can obtain that, for any positive integer $k$ and $1 \leq i \leq k$, if $\textbf{x}_i \in \mathcal{R}^n_\theta$, $(\sum_{i=1}^ka_i\textbf{x}_i)^r = \sum_{i=1}^ka_i\textbf{x}_i^r$, where $a_i\in \mathcal{R}_\theta$.
\label{reverse remark}
\end{remark}
\begin{lemma}
For any positive integer $k$ and $1 \leq i \leq k$, if $\textbf{x}_i \in \mathcal{R}^n_\theta$, then $\phi^{-1}(\phi(\sum_{i=1}^ka_i\textbf{x}_i)^r) = \sum_{i=1}^ka_i\phi^{-1}(\phi(\textbf{x}_i)^r)$, where $a_i\in \mathcal{R}_\theta$.
\label{reverse_l}
\end{lemma}
\begin{proof}
For any $x\in\mathcal{R}_\theta$, $\phi(x)^r$ = $\phi(3x)$, and therefore, for any $\x$ = $(x_1\ x_2\ \ldots\ x_n)\in\mathcal{R}^n_\theta$, $\phi(\x)^r$ = $\phi(3\x^r)$, where $\x^r$ = $(x_n\ x_{n-1}\ \ldots\ z_1)\in\mathcal{R}^n_\theta$. 
Now, 
\begin{equation*}
    \begin{split}
        \phi^{-1}\left(\phi\left(\sum_{i=1}^ka_i\textbf{x}_i\right)^r\right) & = \phi^{-1}\left(\phi\left(\sum_{i=1}^k3a_i\textbf{x}_i^r\right)\right) \\
        & = \sum_{i=1}^k3a_i\textbf{x}_i^r \\
        & = \sum_{i=1}^ka_i(3\textbf{x}_i^r) \\
        & = \sum_{i=1}^ka_i\phi^{-1}\left(\phi\left(\textbf{x}_i\right)^r\right) 
    \end{split}
\end{equation*}
Hence, it follows the result. 
\end{proof}
 
% \begin{lemma}
% For any $\textbf{x}, \textbf{y}  \in \mathcal{R}^n_\theta$ and for any $a,b \in \mathcal{R}_\theta$, $\phi^{-1}(\phi(a\x +b\y)^c) = a\phi^{-1}(\phi(\x)^c) + b\phi^{-1}(\phi(\y)^c)$, 
% \begin{itemize}
%     \item if $a+b \in \{1, 3, 1+2w, 3+2w\}$ when $\theta \in \{2, 2+2w, 1+w, 3+w, 1+3w, 3+3w \}$ and
%     \item if $a+b \in \{1, 3, w, 3w, 2+w, 1+2w, 3+2w, 2+3w\}$ when  $\theta \in \{3, 1+2w \}$.
% \end{itemize} 
% \end{lemma}

% \begin{proof}
% \textbf{The proof will be similar as given for the special case for $\theta = 2+2w$ in Corollary 4.15, \cite{limbachiyaDNA}.}
% \end{proof}
% For higher cases, using similar approach we obtain the following remark.
% \begin{remark}
% For any positive integer $k$ and $1 \leq i \leq k$, if $\textbf{x}_i \in \mathcal{R}^n_\theta$, then $\phi^{-1}(\phi(\sum_{i=1}^ka_i\textbf{x}_i)^c) = \sum_{i=1}^k a_i\phi^{-1}(\phi(\textbf{x}_i)^c)$, 
% \begin{itemize}
%     \item if $\sum_{i=1}^k a_i \in \{1, 3, 1+2w, 3+2w\}$ when $\theta \in \{2, 2+2w, 1+w, 3+w, 1+3w, 3+3w \}$ and
%     \item if $\sum_{i=1}^k a_i \in \{1, 3, w, 3w, 2+w, 1+2w, 3+2w, 2+3w\}$ when $\theta \in \{3, 1+2w \}$.
% \end{itemize} 
% \end{remark}

\begin{lemma}
For a given $\theta$ and for any row $\x \in G_\theta$ over the ring $\mathcal{R}_\theta$, the DNA code $\phi(\langle G_\theta \rangle)$ satisfies the reverse constraint if and only if $\phi^{-1}(\phi(x)^r) \in \langle G_\theta \rangle$, the row span of $G_\theta$ over $\mathcal{R}_\theta$.
\label{reverse}
\end{lemma}

\begin{proof}
Using Remark \ref{reverse_c} and Lemma \ref{reverse_l}, the proof for general case will be similar to the proof given in Lemma 2, \cite{DKBG}, where $\theta =2+2w$ is considered.
\end{proof}

\begin{lemma}
For a given $\theta$ and for any row $\x \in G_\theta$ over the ring $\mathcal{R}_\theta$, the DNA code $\phi(\langle G_\theta \rangle)$ is complement if and only if $\mathbf{2+2w} \in \langle G_\theta \rangle$, where $\mathbf{2+2w}$ is a codeword in $\mathscr{C}_\theta$ with each entry $2 + 2w$.
\label{complement}
\end{lemma}
For more general case of Lemma \ref{complement}, we have the following remark.

\begin{remark}
For a given $\theta$ and for any row $\x \in G_\theta$ over the ring $\mathcal{R}_\theta$, the DNA code $\phi(\langle G_\theta \rangle)$ is complement if and only if $\bm{\lambda} \in \langle G_\theta \rangle$, where $\bm{\lambda}$ is a codeword in $\mathscr{C}_\theta$ with each entry $\lambda\in\{2, 2w, 2+2w\}$.
\end{remark}

\begin{proof}
Using Remark \ref{complement_c} and Table \ref{Gau_Map}, the proof will be similar as given for the special case for $\theta = 2+2w$ in Lemma 3, \cite{DKBG}.
\end{proof}
% \textbf{The advantages of defining Generalized Gau Map $\phi$ along with the Generalized Gau Distance $d_{G(\theta)}$ over the ring $\mathcal{R}_\theta$, can be utilized in determining the property of the DNA code constructed in the proposed way.}
Using Lemma \ref{reverse_l}, Lemma \ref{reverse} and Lemma \ref{complement}, the following theorem gives an insight to the parameters of the constructed DNA code with specific properties.     

\begin{remark}
For a given $\theta$, suppose $\mathscr{C}_\theta (n, M, d_{G(\theta)})$ is a code over the ring $\mathcal{R}_\theta$ with length $n$, the number of the codewords $M$ and the minimum Gau distance $d_{G(\theta)}$, such that the rows of the generator matrix of $\mathscr{C}_\theta$ satisfies the conditions given in the Lemma \ref{reverse} and Lemma \ref{complement}, then $\phi(\mathscr{C}_\theta)$ is a $\mathscr{C}^\theta _{DNA}(2n, M, d_H)$ DNA code with the length $2n$, the number of the codewords $M$ and the minimum Hamming distance $d_H$. 
The DNA code $\mathscr{C}^\theta _{DNA}$ is reversible and reversible-complement.
\label{isometry}
\end{remark}

Using Remark \ref{isometry}, we propose the following constructions of the families of DNA codes that satisfy the reverse and reverse-complement constraints from the the $r^{th}$-order Reed-Muller type codes over the different classes of rings $\mathcal{R}_\theta$.   

For integers $m$ and $r$ ($0\leq r\leq m$), the Reed Muller code is denoted as $\mathcal{R}(r,m)$. 
In this paper, for integers $m$ and $r$ ($0\leq r\leq m$), we define the Reed Muller Type code $\mathcal{R}(r,m)$ over the ring $\mathcal{R}_\theta$ for any $\theta\in\mathcal{R}_\theta$. 
The DNA codes obtained from the Reed Muller Type code satisfies the reverse and reverse-complement constraints. \\
For given $\theta$ and a zero divisor $z$ in the ring $\mathcal{R}_\theta$, the $r^{th}$-order Reed Muller Type code $\mathcal{R}(r,m)$ over $\mathcal{R}_\theta$ is given by the generator matrix 
\[
G_{r,m} =
\begin{pmatrix}
G_{r,m-1} & G_{r,m-1}    \\ 
0 & G_{r-1,m-1} \\
\end{pmatrix},\ 1 \leq r \leq m-1.
\]
with
\[
 G_{m,m}\ = 
 \left(
 \begin{array}{c}
      G_{m-1,m}  \\
      0\ 0\ldots\ 0\ z
 \end{array}
 \right)
 \]
and $G_{0,m}\ = \left(1\ 1\ldots 1\right)$ the all one matrix with length $2^m$.
\begin{theorem}
For the $r^{th}$-order Reed-Muller code $\mathcal{R}(r,m)$ of length $2^m$ over the ring $R_\theta$ ($\theta\in\{2,2+2w\}$), the DNA code $\mathscr{C}_{DNA}=\phi(\mathcal{R}(r,m))$ satisfies the reverse and reverse-complement constraints. 
The $(n,M,d_H)$ parameters of the DNA code $\mathscr{C}_{DNA}$ are $n$ = $2^{m+1}$, %$M = 16 \times \vert <Z> \vert^m $
\[
M = \left \{ \,
\begin{array}{lll}
2^{4b-3a} & if & z \in \{2w\},\\
2^{4b-2a} & if & z \in \{2,2+2w\}, \\
2^{4b-a} & if & z\in \{w,2+w,3w,2+3w\}, \\
\end{array}
\right.\\
\] and
\[
d_{H} = \left \{ \,
\begin{array}{lll}
2^{m-r+1} & if & z \in \{2w,2,2+2w\},\\
2^{m-r} & if & z \in \{w,2+w,3w,2+3w\}, \\
\end{array}
\right. \\
\] where $a = \sum_{i=0}^{r-1}\binom{m-1}{i}$ and $b = \sum_{i=0}^r\binom{m}{i}$.
\label{RM 1}
\end{theorem}
\begin{proof}
The proof follows from the Remark \ref{isometry}, symmetry in generator matrix, and the following facts.
\begin{itemize}
    \item The total number of rows of the generator matrix $G_{r,m}$ is $\sum_{i=0}^r\binom{m}{i}$.
    \item The total number of rows of the generator matrix $G_{r,m}$ that has the element $z$ is $\sum_{i=0}^{r-1}\binom{m-1}{i}$.
    \item Any two codewords in the $r^{th}$-order Reed Muller code $\mathcal{R}_{r,m}$ over $\mathcal{R}_\theta$ are differ at more then $2^{m-r}-1$ positions.
\end{itemize}
\end{proof}

\begin{theorem}
For the $r^{th}$-order Reed-Muller code $\mathcal{R}(r,m)$ of length $2^m$ over the ring $R_\theta$ ($\theta\in\{1+w,3+w,1+3w,3+3w\}$), the DNA code $\mathscr{C}_{DNA}=\phi(\mathcal{R}(r,m))$ satisfy the reverse and reverse-complement constraints. 
The $(n,M,d_H)$ parameters of the DNA code $\mathscr{C}_{DNA}$ are $n$ = $2^{m+1}$, $M$ = $4^{2b-a}$ and $d_H$ = $2^{m-r+1}$, where $a = \sum_{i=0}^{r-1}\binom{m-1}{i}$ and $b = \sum_{i=0}^r\binom{m}{i}$.
\label{RM 2}
\end{theorem}
\begin{proof}
Similar to Theorem \ref{RM 1}, the proof follows from the Remark \ref{isometry}, symmetry in generator matrix, and the following facts.
\begin{itemize}
    \item The total number of rows of the generator matrix $G_{r,m}$ is $\sum_{i=0}^r\binom{m}{i}$.
    \item The total number of rows of the generator matrix $G_{r,m}$ that has the element $z$ is $\sum_{i=0}^{r-1}\binom{m-1}{i}$.
    \item Any two codewords in the $r^{th}$-order Reed Muller code $\mathcal{R}_{r,m}$ over $\mathcal{R}_\theta$ are differ at more then $2^{m-r}-1$ positions.
\end{itemize}
\end{proof}

\begin{theorem}
For the $r^{th}$-order Reed-Muller code $\mathcal{R}(r,m)$ of length $2^m$ over the ring $R_\theta$ ($\theta\in\{3,1+2w\}$), the DNA code $\mathscr{C}_{DNA}=\phi(\mathcal{R}(r,m))$ satisfies the reverse and reverse-complement constraints. 
The $(n,M,d_H)$ parameters of the DNA code $\mathscr{C}_{DNA}$ are $n$ = $2^{m+1}$, 
	\[
	M = \left \{ \,
	\begin{array}{lll}
	2^{4b-3a} & if & z \in \{2+2w\},\\
	2^{4b-2a} & if & z \in \{2,2w\}, \\
	2^{4b-a} & if & z\in \{1+w,3+w,1+3w,3+3w\}, \\
	\end{array}
	\right.\\
	\] and
	\[
	d_{H} = \left \{ \,
	\begin{array}{lll}
	2^{m-r+1} & if & z \in \{2,2w,2+2w\},\\
	2^{m-r} & if & z \in \{1+w,3+w,1+3w,3+3w\}, \\
	\end{array}
	\right. \\
	\] where $a = \sum_{i=0}^{r-1}\binom{m-1}{i}$ and $b = \sum_{i=0}^r\binom{m}{i}$.
\end{theorem}
\begin{proof}
Similar to Theorem \ref{RM 1} and Theorem \ref{RM 2}, the proof follows from the Remark \ref{isometry}, symmetry in generator matrix, and the following facts.
\begin{itemize}
    \item The total number of rows of the generator matrix $G_{r,m}$ is $\sum_{i=0}^r\binom{m}{i}$.
    \item The total number of rows of the generator matrix $G_{r,m}$ that has the element $z$ is $\sum_{i=0}^{r-1}\binom{m-1}{i}$.
    \item Any two codewords in the $r^{th}$-order Reed Muller code $\mathcal{R}_{r,m}$ over $\mathcal{R}_\theta$ are differ at more then $2^{m-r}-1$ positions.
\end{itemize}
\end{proof}

\section{Dual and Self Dual Codes over the rings $\mathcal{R}_{\theta}$}
\label{sec:dual and selfdual}
In this section we study certain specifications of dual codes and establish sufficient conditions of self dual codes which satisfy the combinatorial constraints over the ring $\mathcal{R}_\theta$.  
% and we construct optimal code over $\mathcal{R}$ from the self dual code, which satisfies the reverse and reverse-complement constraints with respect to the Gau map. Furthermore, we discuss the geometrical significance of these self dual codes.\\

For the rest of the paper, Standard Inner Product or Euclidean Inner Product on the ring $\mathcal{R}_\theta$ has been considered. For two vectors $\mathbf{x}= (x_1 \ x_2 \ \ldots \ x_n)$ and $\mathbf{y}= (y_1 \ y_2 \ \ldots \ y_n)$, where $\mathbf{x}, \mathbf{y} \in \mathcal{R}^n _\theta$, the inner product is defined as $[\mathbf{x}, \mathbf{y}]= \sum\limits_{i=1}^n x_i y_i$ over the ring $\mathcal{R}_\theta$.

The dual code $\mathscr{C}^\perp _\theta$ of the code $\mathscr{C}_\theta$ over the ring $\mathcal{R}_\theta$ is defined as the set of annihilators of each of the codewords of the code $\mathscr{C}_\theta$, $i.e., \mathscr{C}^\perp _\theta = \{\mathbf{x} \in \mathcal{R}^n_\theta : [\mathbf{x}, \mathbf{y}] = 0 , \forall \mathbf{y} \in \mathscr{C}_\theta \}$.
A code $\mathscr{C}_\theta$ over the ring $\mathcal{R}_\theta$ will be called self dual, whenever $\mathscr{C}_\theta = \mathscr{C}^\perp _\theta$.

% \begin{lemma}
% For a matrix $G$ over the ring $\mathcal{R}$, the DNA code $\phi(\langle G\rangle_\mathcal{R})$ is complement if and only if $\textbf{2+2w}\in \langle G \rangle_\mathcal{R}$, where $\textbf{2+2w}$ is a codeword in $\mathscr{C}$ with each entry $2+2w$.  
% \end{lemma}

% \begin{theorem}\label{self_orthogonal_cond for Z4}
% A linear code $\mathscr{C}$ over the ring $\mathcal{R}$ is self-orthogonal if and only if each generator matrix $G$ of the code  $\mathscr{C}$ satisfies the following.
% \begin{enumerate}
% \item[a)] Each row of generator matrix $G$ has $\Omega(w)+\Omega(2+w)+\Omega(3w)+\Omega(2+3w) \equiv 0 \pmod 2$ and $(\Omega(1)+\Omega(3)+\Omega(1+2w)+\Omega(3+2w))+2(\Omega(w)+\Omega(2+w)+\Omega(3w)+\Omega(2+3w))+3(\Omega(1+w)+\Omega(3+w)+\Omega(1+3w)+\Omega(3+3w))\equiv 0 \pmod 4$  
% \item[b)] Every pair of rows of the generator matrix $G$ is orthogonal,
% \end{enumerate}
% where $\Omega(x)$ is the occurrence of element $x\in \mathcal{R}$ in the particular row.
% \end{theorem}
% \begin{proof}
% For $x\in \mathcal{R}$, 
% \[
% x^2 = \left\{ \begin{array}{ll}
% 0 & x=0,2,2w,2+2w, \\
% 2+2w & x= w,3w, 2+w, 2+3w, \\
% 1 & x=1,3,1+2w,3+2w,\\
% 3 & x= 1+w,3+w, 1+3w, 3+3w.
% \end{array}\right.
% \]
% If $\Omega(x)$ be the occurrence of the ring element $x$ in vector $\textbf{x}\in \mathcal{R}^n$ then $\textbf{x}.\textbf{x}$  = $0$ if and only if $(2+2w)$($\Omega(w)+\Omega(3w)+\Omega(2+w)+\Omega(2+3w)$)+($\Omega(1)+\Omega(3)+\Omega(1+2w)+\Omega(3+2w)$)+$3$($\Omega(1+w)+\Omega(3+w)+\Omega(1+3w)+\Omega(3+3w)) \equiv 0\pmod 4$.
% \end{proof}

For the linear codes $\mathscr{C}_\theta$ over $\mathcal{R}_\theta, \theta \in \{1+w, 3+w, 1+3w, 3+3w\}$, $i.e.$, of type $\{k_0,k_1\}$, the dual code $\mathscr{C}^\perp _\theta$ will be of type $\{n-k_0-k_1, k_1\}$.
Let $\mathscr{C}_\theta$ be a linear code over $\mathcal{R}_\theta, \theta \in \{2, 2+2w, 3, 1+2w\}$, $i.e.$, $\mathscr{C}_\theta$ is of type $\{k_0,k_1,k_2,k_3\}$. Then the dual code $\mathscr{C}^\perp _\theta$ will be of type $\{n-k_0-k_1-k_2-k_3, k_3,k_2,k_1\}$. 
Hence the following two remarks are obvious.   

\begin{remark}
 If $\mathscr{C}_\theta$ is self dual over $\mathcal{R}_\theta, \theta \in \{1+w, 3+w, 1+3w, 3+3w\}$, then $2k_0 + k_1 = n$.
\end{remark}

\begin{remark}
 If $\mathscr{C}_\theta$ over $\mathcal{R}_\theta, \theta \in \{2, 2+2w, 3, 1+2w\}$ is a self dual code, then $k_1= k_3$ and $2k_0 + 2k_1 + k_2 = n$.
\end{remark}
We denote the parity check matrix of $\mathscr{C}_\theta$ or the generating matrix of $\mathscr{C}^\perp _\theta$ as $G^\perp _\theta$.   
Considering the generating matrix of a linear code $\mathscr{C}_\theta$ in standard form, over the ring $\mathcal{R}_\theta$, the generating matrix of the dual code $\mathscr{C}^\perp _\theta$ is given by:
\begin{itemize}
    \item[a)] For the dual code $\mathscr{C}^\perp _\theta$ over the ring $\mathcal{R}_\theta, \theta \in \{1+w, 3+w, 1+3w, 3+3w\}$, suppose, $g_{_1} = n-k_0-k_1$. Then,
    
  \begin{equation}
G^\perp _{\theta}=\left(
\begin{array}{ccc}
C^T _{0,1} & C^T _{0,2}   & I_{g_{_1}}  \\
-2C^T _{1,1} & 2I_{k_1} &  0  \\

\end{array} \right),
\end{equation} 

  where $C_{0,1}=(A_{0,1}A_{1,2}-A_{0,2})$ and $C_{0,2}=(-A_{1,2}), C_{1,1}=(A_{0,1})$.  
  
  Note that, $A_{i,j}, 0 \leq i < j \leq 2$ are in the form of the generating matrix described in (\ref{G_1}).

     \item[b)] For the dual code $\mathscr{C}^\perp _\theta$ over the ring $\mathcal{R}_\theta, \theta \in \{2, 2+2w\}$, suppose, $g_{_2} = n-k_0-k_1-k_2-k_3$. Then,
  \begin{equation}
G^\perp _{\theta}=\left(
\begin{array}{ccccc}
C^T _{0,1} & C^T _{0,2}   & C^T _{0,3}   & C^T _{0,4}   & I_{g_{_2}}  \\
-wC^T _{1,1}       & wC^T _{1,2} & -wC^T _{1,3} & wI_{k_3} & 0 \\
2C^T _{2,1}      & -2C^T _{2,2}         & 2I_{k_2}  & 0  & 0 \\
-2wC^T _{3,1}       & 2wI_{k_1}        & 0         & 0  & 0 \\
\end{array} \right),
\end{equation} 
  
  where, $C_{0,4} = (-A_{3,4}), C_{0,3}= (A_{2,3}A_{3,4} - A_{2,4}), C_{0,2} = (-A_{1,2}(A_{2,3}A_{3,4} - A_{2,4})+ A_{1,3}A_{3,4} - A_{1,4}), C_{0,1} = (A_{0,1}(- C_{0,2}) - A_{0,2}(A_{2,3}A_{3,4} - A_{2,4}) + A_{0,3}A_{3,4} - A_{0,4})$; $C_{1,1} = (A_{0,1}(A_{1,2}A_{2,3} - A_{1,3}) +A_{0,2}A_{2,3} - A_{0,3}), C_{1,2} = (A_{1,2}A_{2,3} - A_{1,3}), C_{1,3} = (A_{2,3})$; 
  $C_{2,1} = (A_{0,1}A_{1,2} - A_{0,2}), C_{2,2} = (A_{1,2})$; $C_{3,1} = (A_{0,1})$.  
  
  Note that, $A_{i,j}, 0 \leq i < j \leq 4$ are in the form of the generating matrix described in (\ref{G_2}).
  
   \item[c)] For the dual code $\mathscr{C}^\perp _\theta$ over the ring $\mathcal{R}_\theta, \theta \in \{3, 1+2w\}$, suppose, $g_{_3} = n-k_0-k_1-k_2-k_3$. Then, the corresponding generating matrix $G^\perp_\theta$ is given in (\ref{G perp}) with additional conditions, 
   \begin{equation}
G^\perp _{\theta}=\left(
\begin{array}{ccccc}
C^T _{0,1} & C^T _{0,2}   & C^T _{0,3}   & C^T _{0,4}   & I_{g_{_3}}  \\
-(1+w)C^T _{1,1}       & (1+w)C^T _{1,2} & -(1+w)C^T _{1,3} & (1+w)I_{k_3} & 0 \\
2C^T _{2,1}      & -2C^T _{2,2}         & 2I_{k_2}  & 0  & 0 \\
-(2+2w)C^T _{3,1}       & (2+2w)I_{k_1}        & 0         & 0  & 0 \\
\end{array} \right),
\label{G perp}
\end{equation} 
  $C_{0,4} = (-A_{3,4}), C_{0,3}= (A_{2,3}A_{3,4} - A_{2,4}), C_{0,2} = (-A_{1,2}(A_{2,3}A_{3,4} - A_{2,4})+ A_{1,3}A_{3,4} - A_{1,4}), C_{0,1} = (A_{0,1}(- C_{0,2}) - A_{0,2}(A_{2,3}A_{3,4} - A_{2,4}) + A_{0,3}A_{3,4} - A_{0,4})$; $C_{1,1} = (A_{0,1}(A_{1,2}A_{2,3} - A_{1,3}) +A_{0,2}A_{2,3} - A_{0,3}), C_{1,2} = (A_{1,2}A_{2,3} - A_{1,3}), C_{1,3} = (A_{2,3})$; 
  $C_{2,1} = (A_{0,1}A_{1,2} - A_{0,2}), C_{2,2} = (A_{1,2})$; $C_{3,1} = (A_{0,1})$.  
  
  Note that, $A_{i,j}, 0 \leq i < j \leq 4$ are in the form of the generating matrix described in (\ref{G_3}).
\end{itemize}

The existence of self dual codes of any length has been ensured in \cite{doi:10.1504/IJICoT.2010.032133} for the Frobenius Rings. 
Now, we give a construction for all the $16$ rings $\mathcal{R}_\theta$.
For the special cases of $\mathcal{R}_i$, where $i=1,2w$, readers are referred to \cite{bandi2015self}. 

For any given $\theta$ and any positive integer $n$, note that the code $\mathscr{C}_\theta$ = $\{(x\ x\ldots x): x\in A\}$ of length $n$ is a self dual code, where
\[
A=\left\{
\begin{array}{ll}
    \mathcal{R}_\theta & \mbox{if } 4\mid n \\
    \mathcal{Z}_\theta & \mbox{if } 4\nmid n \mbox{ and } 2 \mid n \\
    \{0,2,2w,2+2w\} & \mbox{otherwise}
\end{array}
\right.
\]
and $\mathcal{Z}_\theta$ is the set of all zero divisors for the ring $\mathcal{R}_\theta$.
Therefore, one can observe the following remark.
\begin{remark}
There exist self dual codes of all lengths over $\mathcal{R}_\theta$. 
\end{remark}

\begin{theorem}
For any self-dual code $\mathscr{C}_\theta$ of length $n$ over $\mathcal{R}_\theta$  
\begin{itemize}
    \item[i)] at least one of the codewords $\mathbf{2}$, $\mathbf{2w}$ or $\mathbf{2+2w}$ will be present in $\mathscr{C}_\theta$ over the ring $\mathcal{R}_\theta, \theta \in \{1+w, 3+w, 1+3w, 3+3w \}$;  
    \item[ii)] the codeword $\mathbf{2w}$ will be present in $\mathscr{C}_\theta$ over the ring $\mathcal{R}_\theta, \theta \in \{2, 2+2w\}$;
    \item[iii)] the codeword $\mathbf{2+2w}$ will be present in $\mathscr{C}_\theta$ over the ring $\mathcal{R}_\theta, \theta \in \{3, 1+2w\}$;
\end{itemize}

where $\mathbf{x}$ denotes the all $x$ codeword $(x\ x \ldots x)$ for $x \in \mathcal{R}_\theta$.
\end{theorem}

\begin{proof}
Consider the partition of $\mathcal{R}_\theta$, for any $\theta\in \{1+w, 3+w, 1+3w, 3+3w,2, 2+2w,3, 1+2w \}$, into four sets, $A_0, A_1, A_{w^2}$ and $A_{(1+w)^2}$, where $A_i = \lbrace x : x\in\mathcal{R}, x^2=i\rbrace$.
Now for any $x\in \mathcal{R}_\theta$,
\begin{equation*}
x^2 = \left\{ \begin{array}{ll}
0 & \mbox{ if }x \in \{a+wb: a,b \in 2\mathbb{Z}_4\} \\
1 & \mbox{ if }x \in \{a+wb: a \in 2\mathbb{Z}_4+1, b \in 2\mathbb{Z}_4\} \\
w^2 & \mbox{ if }x \in \{a+wb: a \in 2\mathbb{Z}_4, b \in 2\mathbb{Z}_4+1\} \\
1+2w+w^2 & \mbox{ if }x \in \{a+wb: a,b \in 2\mathbb{Z}_4+1\}.
\end{array}\right.
\end{equation*}
Let $\mathbf{c}= (c_1 \ c_2 \ldots c_n)$ be a codeword in the self dual code $\mathscr{C}_\theta$ of length $n$ over $\mathcal{R}_\theta$. 
Suppose $n_x( \mathbf{c})$ denotes the number of components of a codeword $\mathbf{c}$ which are from the set $A_i$, $i \in \{0, 1, w^2, (1+w)^2\}$. Since $\mathscr{C}_\theta$ is self dual, $\langle \mathbf{c},\mathbf{c} \rangle = 0, \forall \ \mathbf{c} \in \mathscr{C}_\theta$.
Hence, $\sum\limits_{i=1} ^n c_i^2 = 0$.
From this we obtain the following equation. 
\begin{equation}
    (n_1 (\mathbf{c}) + n_{(1+w)^2} (\mathbf{c})) + w^2 (n_{w^2} (\mathbf{c}) + n_{(1+w)^2} (\mathbf{c})) + 2w (n_{(1+w)^2} (\mathbf{c})) = 0
    \label{sd_eq}
\end{equation}
Now, we classify this approach as following.
\begin{itemize}
    \item[i)] $\mathcal{R}_{\theta}, \theta \in \{1+w, 3+w, 1+3w, 3+3w \}=\{a_1+wb_1:a_1,b_1 \in 2\mathbb{Z}_4+1\}$: \\
    In this case, $w^2 = a_1+wb_1$.
    Then equation (\ref{sd_eq}) becomes,
    \begin{equation*}
    \begin{split}
        (n_1 (\mathbf{c}) + a_1(n_{w^2} (\mathbf{c})) & + (a_1+1)(n_{(1+w)^2} (\mathbf{c}))) + \\ 
       &  w((b_1+2)(n_{(1+w)^2} (\mathbf{c})) + b_1(n_{w^2} (\mathbf{c}))) = 0.
    \end{split}
    \end{equation*}
    Hence, we obtain the following two equations,
        \begin{equation}
            n_1 (\mathbf{c}) + a_1(n_{w^2} (\mathbf{c})) + (a_1+1)(n_{(1+w)^2} (\mathbf{c})) = 0,
        \label{eq_1}
        \end{equation}
        \begin{equation}
            (b_1+2)(n_{(1+w)^2} (\mathbf{c})) + b_1(n_{w^2} (\mathbf{c})) = 0.
        \label{eq_2}
        \end{equation}
    Since $b_1$ can be either $1$ or $3$, simplifying equation (\ref{eq_2}), we get the immediate result, $n_{(1+w)^2} (\mathbf{c}) = n_{w^2} (\mathbf{c})$ on $\mathcal{R}_{\theta}$.
    Using similar approach and the previous result, from equation (\ref{eq_1}) we obtain, $n_1 (\mathbf{c}) = n_{w^2} (\mathbf{c})$ on $\mathcal{R}_{\theta}$.  Thus 
    \begin{equation}
        n_1 (\mathbf{c}) = n_{(1+w)^2} (\mathbf{c}) = n_{w^2} (\mathbf{c})
        \label{kg eq-1}
    \end{equation}
    
    Now, using the multiplication table of $\mathcal{R}_\theta$, for any $x\in\mathcal{R}_\theta$, one can observe
    \begin{equation*}
2x = \left\{ \begin{array}{ll}
0 & \mbox{ if }x \in \{a+wb: a,b \in 2\mathbb{Z}_4\} \\
2 & \mbox{ if }x \in \{a+wb: a \in 2\mathbb{Z}_4+1, b \in 2\mathbb{Z}_4\} \\
2w & \mbox{ if }x \in \{a+wb: a \in 2\mathbb{Z}_4, b \in 2\mathbb{Z}_4+1\} \\
2+2w & \mbox{ if }x \in \{a+wb: a,b \in 2\mathbb{Z}_4+1\}
\end{array}\right.
\end{equation*}
So, the inner product of any codeword $\mathbf{c}= (c_1 \ c_2 \ldots c_n)$ with $ \mathbf{2}$ will give, 
\begin{equation*}
    \begin{split}
        \sum_{j=1}^n2c_j &= 0(n_0 (\mathbf{c})) + 2(n_1 (\mathbf{c})) + 2w (n_{w^2} (\mathbf{c})) + (2+2w) (n_{(1+w)^2} (\mathbf{c})) \\ 
        &= 2(n_1 (\mathbf{c}) + n_{(1+w)^2} (\mathbf{c})) + 2w (n_{w^2} (\mathbf{c}) + n_{(1+w)^2} (\mathbf{c})) \\
        &= 0\hspace{0.5cm} (\mbox{from Equation (\ref{kg eq-1})}).
    \end{split}
\end{equation*}
    Therefore, $ \mathbf{2} \in \mathscr{C}^\perp _\theta$, $i.e.$, $\mathbf{2} \in \mathscr{C}_\theta$, since $\mathscr{C}_\theta$ is self-dual.
    From the chain condition over $\mathcal{R}_\theta$, we have, $\langle \mathbf{2} \rangle = \langle \mathbf{2w} \rangle = \langle \mathbf{2w} \rangle$.
    Hence if $\mathscr{C}_\theta$ is self-dual, then the code will code will contain at least one of the codeword $\mathbf{2}$, $\mathbf{2w}$ or $\mathbf{2+2w}$. 
\end{itemize}
    Similar approaches lead to the outcome for the classes of rings $\mathcal{R}_\theta, \theta \in \{2, 2+2w\}$ and $\theta \in \{3, 1+2w\}$ given in ii) and iii), respectively and hence, we obtain the desired result. 
\end{proof}

In the following theorems, the conditions are provided for the construction of self dual codes considering specific generating matrix over the rings $\mathcal{R}_\theta$ that satisfy the reverse and reverse-complement constraints. 

\begin{theorem}
Let $\mathscr{C}_\theta$ be a self dual code over the ring $\mathcal{R}_\theta$ whose generating matrix is $G_\theta= (A_n | B_n)$, where

\begin{equation*}
A_n=\left(
\begin{array}{cccc}
u & 0  & \ldots & 0  \\
0 & u  & \ldots & 0  \\
\vdots & \vdots & \ddots  & \vdots \\
0 & 0 & \ldots & u  \\
\end{array} \right), u \in \mathcal{U}_\theta \text{ and }B_n=\left(
\begin{array}{ccccc}
a_1 & a_2 & \ldots & a_n  \\
a_n & a_1 & \ldots & a_{n-1}  \\
\vdots & \vdots & \ddots  & \vdots \\
a_2 & a_3 & \ldots & a_1  \\
\end{array} \right), 
\end{equation*}

% \begin{equation*}
% B_n=\left(
% \begin{array}{ccccc}
% a_1 & a_2 & a_3 & \ldots & a_n  \\
% a_n & a_1 & a_2 & \ldots & a_{n-1}  \\
%  &  &  &  & \\
% a_2 & a_3 & a_4 & \ldots & a_1  \\
% \end{array} \right), 
% \end{equation*} 
a pure circulant matrix of order $n$.
Then each $a_i$ will also be a unit in $\mathcal{R}_\theta$.
\label{matrix}
\end{theorem}

\begin{proof}
Since $\mathscr{C}_\theta$ is self dual, $u^2 + \sum\limits_{i=1} ^n a_i ^2 = 0$.
Also each row of $G_\theta$ will be orthogonal to each other.
Hence, 

\begin{equation*}
2 \sum\limits_{j=1}^n\sum_{\substack{k=1 \\ j<k }}^n a_j\cdot a_k = 0.    
\end{equation*}

Without loss of generality if we consider, $u$ is an unit in $\mathcal{R}_{2+2w}$ then $u^2$ will be either $3$ or $1$.
If $u^2=3$ then, $\sum\limits_{i=1} ^n a_i ^2 = 1 \implies (\sum\limits_{i=1}^n a_i ^2)^2 = \sum\limits_{i=1}^n a_i^2 = 1$. 
Now there are exactly $4$ solutions of this equation $x^2 = 1$ in $\mathcal{R}_{2+2w}$, i.e., $a_i \in \{1, 3, 1+2w, 3+2w\}$. 
The same result can be obtained by \textbf{similar} argument, for $u^2 = 1$ in $\mathcal{R}_\theta$.
Proofs for other rings $\mathcal{R}_\theta$, $\theta \in \mathbb{Z}_4+w\mathbb{Z}_4 \setminus \{2+2w\}$ will be the same.
Hence, we obtain the desired result.
\end{proof}

\begin{theorem}
The generating matrix $G_\theta$ of a self dual code $\mathscr{C}_\theta$, considered in Theorem \ref{matrix}, will be closed under reverse constraint if $a_i \in \mathcal{U}_\theta \ \forall i \in \{1,2,\ldots n\}$.
\label{Reverse Constraint}
\end{theorem}

\begin{proof}
Recall the generating matrix $G_\theta$ of $\mathscr{C}_\theta$ in Theorem \ref{matrix}. 
Now the reverse of the first row of $G_\theta$ will be $(a_n \ a_{n-1} \ldots a_1  \ldots 0 \ldots u)$.
By inspection we can say that, $a_n = \alpha\cdot u$, where $\alpha \in \mathcal{R}_\theta$. 
So we obtain, $a_i = \alpha_i u ,\ i= 1,2,\ldots n$. 
Now each $\alpha_i$ will be a unit in $\mathcal{R}_\theta$, as $a_i$ and $u$ are units in $\mathcal{R}$. 
Therefore, $a_i = u_i$, where $u_i = \alpha_i u , \ i= 1,2,\ldots n$ and hence the result follows.     
\end{proof}

Using the same argument, as applied in Theorem \ref{Reverse Constraint} the following theorem is obvious. 

\begin{theorem}
The generating matrix $G_\theta$ of a self dual code $\mathscr{C}_\theta$, considered in Theorem \ref{matrix}, will be closed under reverse-complement constraint, if $a_i \in 3\mathcal{U}_\theta \ \forall i \in \{1,2,\ldots n\}$.
\end{theorem}

\begin{proof}
The proof is straightforward.
\end{proof}

\begin{remark}
As $3 \in \mathcal{U}_\theta$ in $\mathcal{R}_\theta$, $G_\theta$ will be closed under reverse-complement constraint if $a_i \in \mathcal{U}_\theta \ \forall i \in \{1,2,\ldots n\}$.
\end{remark}

\section{Bounds for the codes over the ring $\mathcal{R}_\theta$}
\label{sec:Bounds}
In this section we obtain Sphere Packing-like bound, Gilbert Vershamov-like bound, Singleton-like bound and Plotkin-like bound for the codes over the ring $\mathcal{R}_\theta$. 
 
\subsection{Sphere Packing-like Bound on the codes over $\mathcal{R}_\theta$}
In order to establish Sphere Packing-like bound, first we have defined the sphere in Definition \ref{sphere definition} and the establish the Lemma \ref{intersection} and Lemma \ref{circle count}.
\begin{definition}
For any vector $\mathbf{c}$ in $\mathcal{R}^n_\theta$ and any integer $r\geq 0$, the sphere of radius $r$ and center $\mathbf{c}$, denoted as $S_r(\mathbf{c})$, is the set $ \{\mathbf{x} \in \mathcal{R}^n_\theta:d_{G(\theta)} (\mathbf{c}, \mathbf{x}) \leq r \}$. 
\label{sphere definition}
\end{definition}
For any $\textbf{c}\in \mathcal{R}^n_\theta$, if the set $C_r(\mathbf{c})=\{\mathbf{x} \in \mathcal{R}^n_\theta: d_{G(\theta)} (\mathbf{c}, \mathbf{x}) = r \}$, then 
\begin{equation}
    |S_r(\textbf{c})|=\sum_{i=0}^r|C_i(\textbf{c})|.
    \label{sphere count}
\end{equation}

\begin{lemma}
Let $\mathscr{C}_\theta (n,M,d_{G(\theta)})$ be a code over the ring $\mathcal{R}_\theta$. 
Then 
% for $S_r(\mathbf{c}) = \{\mathbf{x} \in \mathcal{R}^n\ |\ d_G (\mathbf{c}, \mathbf{x}) < r \}$, we will have, 
$S_r(\mathbf{c}_1) \cap S_r(\mathbf{c}_2) = \emptyset$, when $r \leq \left \lfloor \frac{d_{G(\theta)} - 1}{2} \right \rfloor, \ \forall \ \mathbf{c}_1, \mathbf{c}_2 \in \mathscr{C}_\theta$.
\label{intersection}
\end{lemma}
\begin{proof}
Consider $S_r(\mathbf{c}_1)$ and $S_r(\mathbf{c}_2)$ for two distinct codewords $\mathbf{c}_1, \mathbf{c}_2 \in \mathscr{C}_\theta$. 
For a vector $\mathbf{x}$ is in both $S_r(\mathbf{c}_1)$ and $S_r(\mathbf{c}_2)$, from the triangle inequality we obtain, 
% % such that $\mathbf{x} \in S_r(\mathbf{c_1}) \cap S_r(\mathbf{c_2})$. Hence,\\
% $d_G (\mathbf{c_1}, \mathbf{x}) < r$ and, $d_G (\mathbf{c_1}, \mathbf{x}) < r$.\\
% Then we will have from the triangular inequality of the Gau distance, 
$d_{G(\theta)} (\mathbf{c}_1, \mathbf{c}_2) \leq d_{G(\theta)} (\mathbf{c}_1, \mathbf{x}) + d_{G(\theta)} (\mathbf{x}, \mathbf{c}_2) \leq r+r= 2r$, 
a contradiction to $d_{G(\theta)} \geq 2r+1$. 
Hence, $S_r(\mathbf{c}_1) \cap S_r(\mathbf{c}_2) = \emptyset$, whenever $r \leq \left \lfloor \frac{d_{G(\theta)} - 1}{2} \right \rfloor, \ \forall \ \mathbf{c}_1, \mathbf{c}_2 \in \mathscr{C}_\theta$.
\end{proof}
\begin{lemma}
For $\textbf{c}\in\mathcal{R}^n_\theta$ and a non-negative integer $r$ $(\leq 2n)$, if $C_r(\textbf{c})=\{\textbf{x}\in\mathcal{R}^n_\theta:d_{G(\theta)}(\textbf{c},\textbf{x})=r\}$, then, %there are  
% A sphere of radius $r$ in $\mathcal{R}$ ($0 \leq r\leq \left \lfloor \frac{d_G - 1}{2} \right \rfloor$) contains exactly 
\begin{equation}
|C_r(\textbf{c})|=\sum\limits_{i= \max\{0,r-n\}}^{\left \lfloor \frac{r}{2} \right \rfloor} \frac{n!}{i!(r-2i)!(n-r+i)!} 9^i6^{r-2i}.
\label{Equation}
\end{equation}
% distinct vectors such that Gau distance between $\textbf{x}$ and the any vector of them is exactly $r$. 
\label{circle count}
\end{lemma}
\begin{proof}
Let $\mathbf{c}$ be a fixed vector in $\mathcal{R}_\theta$. 
We will consider the number of vectors $\mathbf{x}$ that have Gau distance exactly $r$ from $\mathbf{c}$, where $r\leq 2n$. 
From the matrix $\mathscr{M}$ given in (\ref{Gau Distance Matrix}), we can see that the Gau distance between a given element and any distinct element in the same row or column is one, and the Gau distance between the given element and the remaining distinct elements is two. 
Suppose for $\x\in C_r(\textbf{c})$, there are $i$ distinct positions where the elements at those positions at $\x$ and $\textbf{c}$ are two. 
Thus, there are $r-2i$ distinct positions other than the $i$ positions where the elements at those positions at $\x$ and $\textbf{c}$ are one, and elements at remaining postilions $n-r+i$ (= $n-((r-2i)+i)$) are the same.
So, for $i$ elements, which are differ by Gau distance $2$, can be chosen in $9^i$ ways and the other $r-2i$ elements can be chosen in $6^{r-2i}$ ways. 
Now among all these $n$ elements, $i$ and $r-2i$ elements can be arranged in $i!$ and $(r-2i)!$ ways, respectively, among $n$ places. 
The remaining coordinates can be arranged in $(n-r+i)!$ ways. 
So, the total number of vectors in $C_r(\mathbf{c})$ will be,
\begin{equation*}
\sum\limits_{i= \max\{0,r-n\}}^{\left \lfloor \frac{r}{2} \right \rfloor} \frac{n!}{i!(r-2i)!(n-r+i)!} 9^i6^{r-2i}.
\end{equation*}
\end{proof}

\begin{theorem}[Sphere Packing-like Bound]
A code $\mathscr{C}_\theta (n,M,d_{G(\theta)})$ over the ring $\mathcal{R}_\theta$ satisfies,
\begin{equation*}
M\sum_{r=0}^{\left \lfloor\frac{d_{G(\theta)}-1}{2}\right \rfloor}\sum_{i=0}^{\left \lfloor \frac{r}{2} \right \rfloor} \frac{n!9^i6^{r-2i}}{i!(r-2i)!(n-r+i)!} \leq 16^n.
\end{equation*}
\end{theorem}
\begin{proof}
From Lemma \ref{intersection}, any two spheres of radius $m=\left \lfloor\frac{d_{G(\theta)}-1}{2}\right \rfloor$ centred on distinct codewords will be disjoint. 
So, from the Inclusion–Exclusion principle,
\[
\left|\bigcup_{\textbf{c}\in\mathscr{C}_\theta}S_m(\textbf{c})\right|=\sum_{\textbf{c}\in\mathscr{C}_\theta} |S_m(\textbf{c})|.
\]
Therefore, the total number of codewords in each of the $M$ spheres of radius $m$ \textbf{centred} on the $M$ codewords is bounded by $16^n$, $i.e.$
\begin{equation*}
    16^n \geq \left|\bigcup_{\textbf{c}\in\mathscr{C}_\theta}S_m(\textbf{c})\right| = \sum_{\textbf{c}\in\mathscr{C}_\theta} |S_m(\textbf{c})| =\sum_{\textbf{c}\in\mathscr{C}_\theta} \sum_{r=0}^{\left \lfloor\frac{d_{G(\theta)}-1}{2}\right \rfloor} |C_r(\textbf{c})|.
\end{equation*}
But for given non-negative integer $r\ (\leq 2n)$, $|C_r(\textbf{c})|$ is same for each $\textbf{c}\in\mathscr{C}_\theta$. 
Therefore, 
\begin{equation*}
\begin{split}
    16^n \geq M\sum_{r=0}^{\left \lfloor\frac{d_{G(\theta)}-1}{2}\right \rfloor} |C_r(\textbf{c})|.
\end{split}
\end{equation*}
Hence, from Lemma \ref{circle count} and the fact that, for $0\leq r\leq \left \lfloor\frac{d_{G(\theta)}-1}{2}\right \rfloor$, $max\{0,r-n\}=0$, we obtain, 
\begin{equation*}
16^n \geq M\sum\limits_{r=0}^{\left \lfloor\frac{d_{G(\theta)}-1}{2}\right \rfloor}\sum_{i=0}^{\left \lfloor \frac{r}{2} \right \rfloor} \frac{n!9^i6^{r-2i}}{i!(r-2i)!(n-r+i)!}
\end{equation*}

\end{proof}
\begin{corollary}
For the linear code $\mathscr{C}_\theta$ of type $\{k_0, k_1,k_2,k_3\}$ over $\mathcal{R}_\theta$ for $\theta\in\{2,3,1+2w,2+2w\}$, the Sphere Packing-like bound will be given by,
\begin{equation*}
\sum\limits_{r=0}^{\left \lfloor\frac{d_{G(\theta)}-1}{2}\right \rfloor}\sum_{i=0}^{\left \lfloor \frac{r}{2} \right \rfloor}\frac{n!9^i6^{r-2i}}{i!(r-2i)!(n-r+i)!} \leq 16^{n-\left(k_0 +\frac{3k_1}{4} +\frac{k_2}{2}+\frac{k_3}{4}\right)}. 
\end{equation*}
\end{corollary}
\begin{corollary}
For the linear code $\mathscr{C}_\theta$ of type $\{k_0, k_1\}$ over $\mathcal{R}_\theta$ for $\theta\in\{1+w,1+3w,3+w,3+3w\}$, the Sphere Packing-like bound will be given by,
\begin{equation*}
\sum\limits_{r=0}^{\left \lfloor\frac{d_{G(\theta)}-1}{2}\right \rfloor}\sum_{i=0}^{\left \lfloor \frac{r}{2} \right \rfloor}\frac{n!9^i6^{r-2i}}{i!(r-2i)!(n-r+i)!} \leq 16^{n-\left(k_0 +\frac{k_1}{2}\right)}.
\end{equation*}
\end{corollary}

\subsection{Gilbert Varshamov-like Bound}
\begin{theorem}[Gilbert Varshamov-like Bound]
If $A_{16}(n,d_{G(\theta)})$ is the maximum possible size of a code $\mathscr{C}_\theta$ with length $n$ and minimum Gau distance $d_{G(\theta)}$ over the ring $\mathcal{R}_\theta$ then
\begin{equation*}
16^n\leq A_{16}(n,d_{G(\theta)}) \sum_{r=0}^{d_{G(\theta)}-1}\sum_{i= \max\{0,r-n\}}^{\left \lfloor \frac{r}{2} \right \rfloor} \frac{n!9^i6^{r-2i}}{i!(r-2i)!(n-r+i)!}.
\end{equation*}
\end{theorem}
\begin{proof}
Let $\mathscr{C}_\theta$ $(n,M,d_{G(\theta)})$ be a code over the ring $\mathcal{R}_\theta$ with maximal size $|\mathscr{C}_\theta|$ = $M$ = $A_{16}(n,d_{G(\theta)})$. 
Therefore, for each $\textbf{x}\in \mathcal{R}^n_\theta$, there exists at least one codeword $\textbf{c}\in\mathscr{C}_\theta$ such that $d_{G(\theta)}(\textbf{c},\textbf{x})\leq d_{G(\theta)}-1$. 
So, \begin{equation*}
\mathcal{R}^n_\theta=\bigcup_{\textbf{c}\in\mathscr{C}_\theta}S_{d_{G(\theta)}-1}(\textbf{c}),    
\end{equation*}
where the sphere $S_{d_{G(\theta)}-1}(\textbf{c})$ is given by Definition \ref{sphere definition}.
But from Lemma \ref{circle count} and Equation \ref{sphere count}, for each $\textbf{c}\in\mathscr{C}_\theta$ the sphere size will be
\begin{equation*}
    \begin{split}
        |S_{d_{G(\theta)}-1}(\textbf{c})| = & \sum_{r=0}^{d_{G(\theta)}-1}|C_r(\textbf{c})| \\
        = & \sum_{r=0}^{d_{G(\theta)}-1}\sum_{i= \max\{0,r-n\}}^{\left \lfloor \frac{r}{2} \right \rfloor} \frac{n!9^i6^{r-2i}}{i!(r-2i)!(n-r+i)!}.
    \end{split}
\end{equation*}
Hence, 
\begin{equation*}
    \begin{split}
        16^n = & |\mathcal{R}^n_\theta| \\ 
        = & \left|\bigcup_{\textbf{c}\in\mathscr{C}_\theta}S_{d_{G(\theta)}-1}(\textbf{c})\right| \\ 
        \leq & \sum_{\textbf{c}\in\mathscr{C}_\theta}\sum_{r=0}^{d_{G(\theta)}-1}\sum_{i= \max\{0,r-n\}}^{\left \lfloor \frac{r}{2} \right \rfloor} \frac{n!9^i6^{r-2i}}{i!(r-2i)!(n-r+i)!} \\
        = & A_{16}(n,d_{G(\theta)})\sum_{r=0}^{d_{G(\theta)}-1}\sum_{i= \max\{0,r-n\}}^{\left \lfloor \frac{r}{2} \right \rfloor} \frac{n!9^i6^{r-2i}}{i!(r-2i)!(n-r+i)!}.
    \end{split}
\end{equation*}
\end{proof}
\subsection{Singleton-like Bound}
To establish the Singleton-like Bound, we first need to prove the following Lemma.

\begin{lemma}
% Let $\mathcal{C}$ $(n,M,d_G)$ be a code over the ring $\mathcal{R}$. Then by deleting the first $\left \lfloor \frac{d_G - 1}{2} \right \rfloor\ (=r)$ letters in every codewords of $\mathcal{C}$, we will obtain a new code $\mathcal{C^*}$ with exactly $M$ numbers of distinct codewords.   
Let $\mathscr{C}_\theta$ $(n,M,d_{G(\theta)})$ be a code over the ring $\mathcal{R}_\theta$. Then by deleting the first $\left \lfloor \frac{d_{G(\theta)} - 1}{2} \right \rfloor\ (=r)$ letters in every codeword of $\mathscr{C}_\theta$, we will obtain exactly $M$ distinct codewords of length $n-\left \lfloor \frac{d_{G(\theta)} - 1}{2} \right \rfloor$, i.e. $n - r$.
\label{deletaed code word}
\end{lemma}

\begin{proof}
Let $\mathbf{x}=( x_1\ x_2\ldots\ x_{r-1}\ x_r\ x_{r+1}\ldots\ x_n)$ and $\mathbf{y}=(y_1\ y_2\ldots\ y_{r-1}\ y_r$ $ y_{r+1}\ldots\ y_n)$ be distinct codewords in $\mathscr{C}_\theta$, such that $(x_{r+1}\ldots\ x_n)$ = $(y_{r+1}\ldots$ $y_n)$.
Then $\exists$ at least one $j\in \{1,2,\ldots ,r\}$ such that $d_{G(\theta)} (x_j , y_j) \neq 0$ (since, $\mathbf{x} \neq \mathbf{y}$). \\
% (\textbf{The statement is not true for some cases} for example if $\mathscr{C}$=$\{11111111,1111100w\}$, then $11111111\neq1111100w$ but $d_G (x_j , y_j) = 0$ for $j\in \{1,2,\ldots ,r\}$.).
Now consider the general case,
% \textbf{Case I:}\ \ $d_G (x_i , y_i) = 1 \ \forall \ 1 \leq i \leq r $. \\
% Then, \ $d_G \leq d_G (\mathbf{x},\mathbf{y}) = \Omega_{i=1}^r d_G (x_i , y_i) = r$.\\
% $\implies \left \lfloor \frac{d_G - 1}{2} \right \rfloor \leq \left \lfloor \frac{r - 1}{2} \right \rfloor < r$.\\
% which is contradiction since, $\left \lfloor \frac{d_G - 1}{2} \right \rfloor =r$. So for this case, our assumption was wrong.\\
% \textbf{Case II:}\ \ $d_G (x_i , y_i) = 2 \ \forall \ 1 \leq i \leq r $. \\
% Then, \ $d_G \leq d_G (\mathbf{x},\mathbf{y}) = \Omega_{i=1}^r d_G (x_i , y_i) = 2r$.\\
% $\implies \left \lfloor \frac{d_G - 1}{2} \right \rfloor \leq \left \lfloor \frac{2r - 1}{2} \right \rfloor < r$.\\
% which is contradiction since, $\left \lfloor \frac{d_G - 1}{2} \right \rfloor =r$. So for this case also, our assumption was wrong.\\
$d_{G(\theta)} (\mathbf{x},\mathbf{y}) = r_1 + 2r_2 $, where $r = r_1 + r_2$ and $r_1,r_2$ are nonnegative integers. 
Then, \ $d_{G(\theta)} \leq r_1 + 2r_2$, $i.e.$,
$\left \lfloor \frac{d_{G(\theta)} - 1}{2} \right \rfloor \leq \left \lfloor \frac{(r_1 + 2r_2) - 1}{2} \right \rfloor$ which contradicts, $r = \left \lfloor \frac{d_{G(\theta)} - 1}{2} \right \rfloor < \left \lfloor \frac{2r - 1}{2} \right \rfloor < r$. 
Hence there will be exactly $M$ distinct codewords of length $n - r$ in $\mathscr{C}_\theta$. 
\end{proof}
\begin{remark}
From Lemma \ref{deletaed code word}, it is obvious that $\exists$ at least one $l \geq r+1$ such that $d_{G(\theta)} (x_l , y_l) \neq 0$.\\
\end{remark}

Using the above lemma we can proceed to the Singleton-like bound for a code over $\mathcal{R}_\theta$.
\begin{theorem}[Singleton-like Bound]
If $\mathscr{C}_\theta (n,M,d_{G(\theta)})$ is a code (not necessarily linear) over $\mathcal{R}_\theta$, then 
$
M \leq 16^{n- \left \lfloor \frac{d_{G(\theta)} - 1}{2} \right \rfloor}.    
$
\label{singleton bound}
\end{theorem}
\begin{proof}
From Lemma \ref{deletaed code word}, we have exactly $M$ distinct codewords of length $n - r$, by deleting the first $r\ (=\left \lfloor \frac{d_{G(\theta)} - 1}{2} \right \rfloor)$ letters in every codewords of $\mathcal{C}_\theta$. 
The order of the ring $\mathcal{R}_\theta$ is sixteen, so $M \leq 16^{n- \left \lfloor \frac{d_{G(\theta)} - 1}{2} \right \rfloor}$ which is the Singleton-like bound for any code $\mathcal{C}_\theta$ (not necessarily linear) over $\mathcal{R}_\theta$.
\end{proof}
\begin{corollary}
If the code $\mathscr{C}_\theta (n,M,d_{G(\theta)})$ is linear and of type $\{k_0,k_1,k_2,k_3\}$ over $\mathcal{R}_\theta$ for $\theta\in\{2,3,1+2w,2+2w\}$, then 
\begin{equation*}
{\left \lfloor \frac{d_{G(\theta)} - 1}{2} \right \rfloor} \leq n - \left( k_0 +\frac{3k_1}{4} +\frac{k_2}{2}+\frac{k_3}{4}\right).     
\end{equation*}
\end{corollary}\begin{corollary}
If the code $\mathscr{C}_\theta (n,M,d_{G(\theta)})$ is linear and of type $\{k_0,k_1\}$ over $\mathcal{R}_\theta$ for $\theta\in\{1+w,3+w,1+3w,3+3w\}$, then 
\begin{equation*}
{\left \lfloor \frac{d_{G(\theta)} - 1}{2} \right \rfloor} \leq n - \left( k_0 +\frac{k_1}{2}\right).     
\end{equation*}
\end{corollary}
\begin{definition}
Any code $\mathscr{C}_\theta$ defined on the ring $\mathcal{R}_\theta$ is called Maximum Gau Distance Separable (MGDS) code, if the code $\mathscr{C}_\theta$ satisfies the Singleton-like bound as given in Theorem \ref{singleton bound}.
\end{definition}
\begin{remark}
For any positive integer $n$ and $\theta\in \mathbb{Z}_4+w\mathbb{Z}_4$, the code $\mathscr{C}_\theta$ = $\mathcal{R}^n_\theta$ with the parameter $(n,16^n,1)$ is MGDS code over the ring $\mathcal{R}_\theta$. 
\end{remark}
\begin{remark}
For any positive integer $n$, if $\textbf{a}\in\{2,2+2w\}^n$ then the code $\mathscr{C}_\theta$ = $\langle \textbf{a}\rangle$ with the parameter $(n,4,2n)$ is MGDS code over the ring $\mathcal{R}_\theta$, where $\theta\in\{2,2+2w\}$. 
\end{remark}
\begin{remark}
For any positive integer $n$, if $\textbf{a}\in\{2,2w\}^n$ then the code $\mathscr{C}_\theta$ = $\langle \textbf{a}\rangle$ with the parameter $(n,4,2n)$ is MGDS code the ring $\mathcal{R}_\theta$, where $\theta\in\{3,1+2w\}$. 
\end{remark}
\begin{remark}
For any positive integer $n$, if $\textbf{a}\in\{2,2w,2+2w\}^n$ then the code $\mathscr{C}_\theta$ = $\langle \textbf{a}\rangle$ with the parameter $(n,4,2n)$ is MGDS code the ring $\mathcal{R}_\theta$, where $\theta\in\{1+w,3+w,1+2w,3+3w\}$. 
\end{remark}
\subsection{Plotkin-like bound on the codes over $\mathcal{R}_\theta$}

In the following we establish the Plotkin-like bound considering the Generalized Gau distance over the rings $\mathcal{R}_\theta$.
\begin{theorem}[Plotkin-like Bound]
A code $\mathscr{C}_\theta$ $(n,M,d_{G(\theta)})$ over the ring $\mathcal{R}_\theta$ will satisfy,
\begin{equation*}
    M \leq \left\lfloor\frac{2d_{G(\theta)}}{2d_{G(\theta)}-3n}\right\rfloor\ \mbox{ for }2d_{G(\theta)}>3n.
\end{equation*}
\label{Plotkin Bound}
\end{theorem}
\begin{proof}
For positive integers $n$ and $M$, and $i=1,2,\ldots,M$, consider $\mathbf{x}_i=(x_{i,1}\ x_{i,2}\ldots x_{i,n})\in\mathscr{C}_\theta$. 
We will prove the theorem by using lower and upper bounds of
$
\sum_{i=1}^M\sum_{j=1}^M d_{G(\theta)} (\mathbf{x}_i,\mathbf{x}_j).
$
For lower bound, there are $M$ choices for $\mathbf{x}_i$ and hence, $M-1$ choices are left for $\mathbf{x}_j$. 
Therefore,
 \begin{equation}
    \sum_{i=1}^M\sum_{\substack{j=1 \\ i\neq j}}^M d_{G(\theta)} (\mathbf{x}_i,\mathbf{x}_j) \geq M(M-1)d_{G(\theta)}. 
 \label{plot eq1}
 \end{equation}
For upper bound, consider a matrix of order $M\times n$, whose rows represent the codewords of $\mathscr{C}_\theta$. 
So, 
\begin{equation}
    \sum_{i=1}^M\sum_{\substack{j=1 \\ i\neq j}}^M d_{G(\theta)} (\mathbf{x}_i,\mathbf{x}_j) = 
    \sum_{l=1}^n\sum_{i=1}^M\sum_{\substack{j=1 \\ i\neq j}}^M d_{G(\theta)} (x_{i,l}, x_{j,l}).
\label{plot inter eq 1}
\end{equation}
For each $m\in\mathcal{R}_\theta$, let $s_{m,l}$ ($l= 1,2, \ldots n$) denotes the occurrence of the ring element $m$ in the $l^{th}$ column of the matrix.
Then for $l= 1,2, \ldots, n$, 
\begin{equation}
    \sum_{m\in\mathcal{R}}s_{m,l}=M.
    \label{plot eq 6}
\end{equation}
Let us also define two sets $A_m=\{y: d_G(m,y)=1, y\in\mathcal{R}_\theta\}$ and $B_m=\{y: d_G(m,y)=2, y\in\mathcal{R}_\theta\}$, for each $m\in\mathcal{R}_\theta$, using properties of Gau distance.
Then $A_m\cup B_m\cup\{m\}=\mathcal{R}_\theta$, for each $m$. 
Now, for $l=1,2,\ldots,n$, if $s_{i,l}$ is the number of occurrence of the element $i\in\mathcal{R}_\theta$ in the column $l$ of the matrix then
\begin{align}
    \sum_{i=1}^M  \sum_{\substack{j=1 \\ i\neq j}}^M d_G (x_{i,l}, x_{j,l}) 
& = 
\sum_{m\in\mathcal{R}_\theta}\binom{s_{m,l}}{1}\left(1\sum_{j\in A_m}\binom{s_{j,l}}{1}+2\sum_{j\in B_m}\binom{s_{j,l}}{1}\right) \label{plot eq 4} \\ 
& =
\sum_{m\in\mathcal{R}_\theta}s_{m,l}\left(\sum_{j\in A_m}s_{j,l}+2\sum_{j\in B_m}s_{j,l}\right).
\end{align}
The symmetricity of the RHS in Equation (\ref{plot eq 4}) implies that it will be maximum along with the constraint Equation (\ref{plot eq 6}) if $s_{i,l}=M/16$ for each $i\in\mathcal{R}_\theta$ and $l=1,2,\ldots,n$.
Hence, 
\begin{equation}
\sum_{m\in\mathcal{R}_\theta} s_{m,l} \left(\sum_{j\in A_m}s_{j,l}\right. +  \left.2\sum_{j\in B_m}s_{j,l}\right) \leq \sum_{m\in\mathcal{R}_\theta}\frac{24}{16^2}M^2=\frac{3M^2}{2},
\label{plot eq 5}
\end{equation}
where $|A_m|=6$ and $|B_m|=9$ for each $m\in\mathcal{R}_\theta$.
From Inequality (\ref{plot eq1}), Equation (\ref{plot inter eq 1}), Equation (\ref{plot eq 4}) and Inequality (\ref{plot eq 5}), we obtain,
$
    M(M-1)d_{G(\theta)}\leq\sum_{l=1}^n\frac{3M^2}{2}.
$
Therefore, 
$
    M \leq \frac{2d_{G(\theta)}}{2d_{G(\theta)}-3n}.
$
Note that $M,n$ and $d_{G(\theta)}$ are finite positive integers.
Therefore, 
$
    M \leq \left\lfloor\frac{2d_{G(\theta)}}{2d_{G(\theta)}-3n}\right\rfloor \mbox{ when } 2d_{G(\theta)}>3n.
$
\end{proof}

\begin{corollary}
If the code $\mathscr{C}_\theta (n,M,d_{G(\theta)})$ is linear and of type $\{k_0,k_1,k_2,k_3\}$ over $\mathcal{R}_\theta$, $\theta\in\{2,3,1+2w,2+2w\}$, then,
\begin{equation*}
 \left( k_0 +\frac{3k_1}{4} +\frac{k_2}{2}+\frac{k_3}{4}\right) \leq log_{16}\left( \frac{2d_{G(\theta)}}{2d_{G(\theta)}-3n} \right)    
\end{equation*}
when $2d_{G(\theta)}>3n$.
% where $2d_G>3n$.
\end{corollary}
\begin{corollary}
If the code $\mathscr{C}_\theta (n,M,d_{G(\theta)})$ is linear and of type $\{k_0,k_1\}$ over $\mathcal{R}_\theta$, $\theta\in\{1+w,1+3w,3+w,3+3w\}$, then for $2d_{G(\theta)}>3n$,
\begin{equation*}
 \left( k_0 +\frac{k_1}{2}\right) \leq \log_{16}\left( \frac{2d_{G(\theta)}}{2d_{G(\theta)}-3n} \right).    
\end{equation*}
% where $2d_G>3n$.
\end{corollary}
\begin{remark}
For any positive integer $n$, if $\textbf{a}\in\{2,2+2w\}^n$ then the family of codes $\mathscr{C}_\theta (n,4,2n)$ = $\langle \textbf{a}\rangle$ satisfies the Plotkin-like bound with equality over the ring $\mathcal{R}_\theta$, where $\theta\in\{2,2+2w\}$. 
\end{remark}
\begin{remark}
For any positive integer $n$, if $\textbf{a}\in\{2,2w\}^n$ then the class of codes $\mathscr{C}_\theta (n,4,2n)$ = $\langle \textbf{a}\rangle$ will be optimal with respect to the Plotkin-like bound over the ring $\mathcal{R}_\theta$, where $\theta\in\{3,1+2w\}$. 
\end{remark}
\begin{remark}
For any positive integer $n$, if $\textbf{a}\in\{2,2w,2+2w\}^n$ then the class of codes $\mathscr{C}_\theta (n,4,2n)$ = $\langle \textbf{a}\rangle$ satisfies the Plotkin-like bound optimally over the ring $\mathcal{R}_\theta$, where $\theta\in\{1+w,3+w,1+2w,3+3w\}$. 
\end{remark}

\section{Conclusion}
\label{sec:Conclusion}
The present paper includes the exhaustive study of the $16$ ring structures $\mathbb{Z}_4+w\mathbb{Z}_4$, $w^2 \in \mathbb{Z}_4+w\mathbb{Z}_4$.
We extended the concept of Gau map and Gau distance, proposed in \cite{DKBG}, over the entire class of $16$ rings.
The isometry between the codes over the rings and the analogous DNA codes has been established universally considering Gau distance and Hamming distance over the elements of the rings and DNA strings of length $2$, respectively.   
In addition to this, necessary and sufficient conditions have been proposed so that the constructed DNA codes satisfy reverse and reverse-complement constraints.
We enrich the study of dual and self dual codes along with the characterization of corresponding generator matrix with different specifications in a general manner over the rings.
Prominent bounds $i.e.$, Sphere Packing-like bound, GV-like bound, Singleton-like bound and Plotkin-like bound have been proposed considering the Gau distance over all the rings.
Parameters for special classes of optimal codes have been proposed that meet Singleton-like bound and Plotkin-like bound, respectively. 
Also the construction of a class of DNA codes using the Reed-Muller type codes over the rings has been given, which satisfies reverse and reverse-complement constraints.
In future, it will be an interesting task to define Gau weight over the rings considered in this work.
Utilizing the advantage, the derivation of McWilliam's Identity can be done in future.
The construction of different families of optimal codes for each of the proposed bounds will also be an engaging work.

\section*{Acknowledgment}
This work is funded by the Deanship of Scientific Research (DSR), King Abdulaziz University, Jeddah under grant number (DF-594-130-1441). The authors, therefore, gratefully acknowledge DSR technical and financial support.

\bibliographystyle{IEEEtran} 
\bibliography{mybibfile}

\end{document}